\documentclass[sigconf,10pt]{acmart}
\usepackage[utf8]{inputenc}
\usepackage{microtype}
\usepackage{breakurl}

\title{Enumeration on Trees with Tractable Combined Complexity and Efficient Updates}

\copyrightyear{2019} 
\acmYear{2019} 
\setcopyright{rightsretained}
\acmConference[PODS '19] {38th ACM SIGMOD-SIGACT-SIGAI Symposium on Principles
of Database Systems}{Extended version}{Appendices included}

\hypersetup{
    colorlinks,
    linkcolor={red!50!black},
    citecolor={blue!50!black},
    urlcolor={blue!30!black}
}
\usepackage{balance} 
\usepackage{multirow}
\usepackage{amsmath}
\usepackage{tikz}
\usetikzlibrary{shapes,snakes}
\usepackage{xspace}
\usepackage{tabularx}
\usepackage[bibliography=common]{apxproof}
\usepackage{bm}
\usepackage{colonequals}
\usepackage{amsthm}
\usepackage{amsfonts}
\usepackage{paralist}
\usepackage{booktabs}
\usepackage[noend]{algpseudocode}
\usepackage{algorithm}
\usepackage{chngcntr}
\makeatletter
\newcommand*{\defeq}{\mathrel{\rlap{%
  \raisebox{0.3ex}{$\m@th\cdot$}}%
  \raisebox{-0.3ex}{$\m@th\cdot$}}%
  =}
\makeatother

\newcommand{\card}[1]{\left|{#1}\right|}

\newcommand{\query}{\Phi}

\renewcommand{\phi}{\varphi}

\newcommand{\calT}{\mathcal{T}}

\newcommand{\X}{\mathcal{X}}
\newcommand{\Y}{\mathcal{Y}}

\renewcommand{\S}{\mathrm{S}}

\newcommand{\Prov}{\mathrm{Prov}}

\newcommand{\var}{\mathrm{var}}

\newcommand{\Leaf}{\mathrm{Leaf}}

\newcommand{\fib}{\textsc{fib}}
\newcommand{\fbb}{\textsc{fbb}}

\newcommand{\node}{\mathsf{node}}

\newcommand{\hole}{\Box}
\newcommand{\boxf}{\textsc{box}}
\newcommand{\leftb}{\textsc{left}}
\newcommand{\rightb}{\textsc{right}}

\newcommand{\transitive}[1]{{\downarrow}\hspace{-.1em}({#1})}

\newcommand{\depth}{\mathrm{depth}\xspace}
\newcommand{\height}{\mathrm{height}\xspace}
\newcommand{\poly}{\mathrm{poly}\xspace}

\newcommand{\cuppath}{\mathrel{\raisebox{-3pt}{$\overset{\cup}{\leadsto}$}}}

\newcommand{\result}[1]{{#1}^\text{res}}

\newtheoremrep{theorem}{Theorem}
\newtheoremrep{lemma}[theorem]{Lemma}
\newtheoremrep{corollary}[theorem]{Corollary}
\counterwithin{theorem}{section}

\algrenewcommand{\algorithmicindent}{1em}

\renewcommand{\appendixprelim}{\onecolumn\newgeometry{margin=4cm}\pagestyle{plain}}
\begin{abstract}
We give an algorithm to enumerate the results on trees of monadic second-order
(MSO) queries 
represented by \emph{nondeterministic} tree automata. After
linear time preprocessing (in the input tree), we can enumerate answers with
linear delay (in each answer). We allow updates on the tree to take place at
any time, and we can then restart the enumeration after logarithmic time in the
tree.
Further, all our combined complexities are polynomial in the automaton.

Our result follows our previous circuit-based enumeration algorithms based on deterministic
tree automata,
and is also inspired by our earlier result on
words and nondeterministic sequential extended variable-set
automata
in the context of document spanners. 
We extend these results and combine them with a
recent tree balancing scheme by Niewerth, so that our enumeration structure
supports updates to the underlying tree in logarithmic time (with leaf
insertions, leaf deletions, and node relabelings).
Our result implies that,
for MSO queries with free first-order variables, we can
enumerate the results with linear preprocessing and constant-delay and update
the underlying tree in logarithmic time, which improves on several known results for words and trees.

Building on lower bounds from data structure research, we also show unconditionally that up to a doubly logarithmic factor the update
time of our algorithm is optimal. Thus, unlike other settings, 
there can
be no algorithm with constant update time.

\end{abstract}

\author{Antoine Amarilli}
\affiliation{LTCI, Télécom ParisTech, Université Paris-Saclay}
\author{Pierre Bourhis}
\affiliation{CRIStAL, CNRS UMR 9189 \& Inria Lille}
\author{Stefan Mengel}
\affiliation{CNRS, CRIL UMR 8188}
\author{Matthias Niewerth}
\affiliation{University of Bayreuth}

\begin{document}
\newtheorem{claim}[theorem]{Claim}
\newtheorem{observation}[theorem]{Observation}
\maketitle

\section{Introduction}
\label{sec:intro}
Query evaluation is one of the central tasks in databases.
A prominent framework for this task in theoretical work
are so-called \emph{enumeration algorithms}. Such algorithms print out
query answers one after the other, without duplicates, and bound the time between each two answers, called the \emph{delay} of the algorithm.
Enumeration algorithms were introduced in database theory by
the pioneering work of Durand and Grandjean~\cite{durand2007first},
in which the goal was to achieve a constant bound on the delay, called
\emph{constant-delay enumeration}.
In such algorithms, the delay between different solutions only depends on the
query to be evaluated,
in particular it is independent of the database.
To ensure such a bound, the algorithm is allowed to first perform a
\emph{preprocessing phase} on the database,
which in most cases is required to run in linear time.
While these time bounds sound very restrictive, 
a surprising range of query evaluation problems allow enumeration
in this regime, as surveyed, e.g., in~\cite{Segoufin14}. Most notably,
constant-delay enumeration with linear-time preprocessing 
can be achieved for acyclic free-connex queries on arbitrary
databases~\cite{bagan2007acyclic} (with extensions for the case of functional
dependencies~\cite{carmeli2018enumeration} and for unions of conjunctive
queries~\cite{berkholz2018answering}), for first-order queries on bounded degree
databases~\cite{berkholz2017answering2} and nowhere dense
databases~\cite{schweikardt2018enumeration} (with pseudo-linear time
preprocessing), for regular path
queries~\cite{martens2018evaluation}, and for monadic second-order queries
on trees and bounded-treewidth structures
\cite{bagan2006mso,kazana2013enumeration} which is the setting that we study.

Enumeration algorithms have recently been extended to handle the
fact that practical databases are rarely static but change frequently. 
Specifically, it is important for enumeration algorithms to efficiently handle
updates to the underlying database, and maintain their precomputed index structures
without recomputing them from scratch at each update.
This approach was inspired by the work of Balmin et al.~\cite{balmin2004incremental}
that showed how to maintain queries on trees efficiently under updates, but only
for Boolean queries, i.e., without the ability to enumerate results.
The work of Losemann and Martens \cite{losemann2014mso} was the first to extend this to
an update-aware enumeration algorithm,
which handled efficient enumeration of queries in
monadic second order logic (MSO) 
queries on words and
trees.
Note that the efficiency in such update-aware enumeration algorithms is measured by three parameters: the enumeration delay, the time for an initial preprocessing, and the time necessary for each update.
Since~\cite{losemann2014mso}, several such algorithms have been designed in
several settings, namely, CQs \cite{berkholz2017answering} and UCQs
\cite{berkholz2018answering}, FO+MOD queries on bounded degree databases~\cite{berkholz2017answering2}, and MSO queries over words and trees (see
Table~\ref{tab:overview}).

In this paper, we focus on the setting where we evaluate 
MSO queries on trees~\cite{bagan2006mso}. This
task is motivated by, e.g., querying tree-shaped data such as XML or JSON
documents, and the important special case of query evaluation on words, e.g.,
for information extraction using \emph{document spanners}~\cite{FaginKRV15,florenzano2018constant}.
It is already known that for any MSO query the answers on a given tree can be enumerated
in linear time~\cite{bagan2006mso,kazana2013enumeration},
and this was extended to updates by Losemann and Martens~\cite{losemann2014mso}
but at the cost of increasing the enumeration delay to a polylogarithmic term
(in particular making it depend on the database).
This delay and update time were recently made logarithmic in
\cite{Niewerth18}, and two incomparable algorithms can achieve constant
delay in special cases: when restricting the structure to be a
word~\cite{NiewerthS18}, or when restricting updates to \emph{relabelings},
i.e., disallowing structural changes on the input tree. See 
Table~\ref{tab:overview} for an overview.
\begin{table}[t]
  {\centering
    \renewcommand{\tabcolsep}{2pt}
  \begin{tabularx}{\linewidth}{Xlll}
    \toprule
    {\bfseries Work} & {\bfseries Data} & {\bfseries Delay} & {\bfseries Updates}
    \\
    \midrule
    Bagan~\cite{bagan2006mso}, \newline Kazana and Segoufin~\cite{kazana2013enumeration} & trees & $O(1)$ & N/A
    \\
    Losemann and Martens~\cite{losemann2014mso} & words & $O(\log n)$ & $O(\log n)$ 
    \\
    Losemann and Martens~\cite{losemann2014mso} & trees & $O(\log^2 n)$ & $O(\log^2 n)$ 
    \\
    Niewerth and Segoufin~\cite{NiewerthS18} & words & $O(1)$ &
    $O(\log n)$ \\
    Niewerth~\cite{Niewerth18} & trees & $O(\log n)$ &
    $O(\log n)$ \\
    Amarilli, Bourhis, Mengel \cite{amarilli2018enumeration} &
    trees & $O(1)$ & $O(\log n)^1$\\
    \midrule
    \bfseries{this paper} & \bfseries{trees} & \bfseries{$\bm{O(1)}$} &
    \bfseries{$\bm{O(\log n)}$}\\
    \bottomrule
  \end{tabularx}}

  \begin{flushleft}$^1$: Only supports \emph{relabeling updates}\end{flushleft}
    \caption{State of the art for enumeration of MSO query results (for free
    FO variables) on trees under updates. All these algorithms have linear time
  preprocessing.\vspace{-.7cm}}\label{tab:overview}
\end{table}

Our first contribution in this paper is to show an update-aware enumeration
algorithm for MSO query evaluation on trees that achieves linear-time
preprocessing, constant-delay enumeration, and supports updates in logarithmic
time including node insertions and deletions. This strictly
outperforms all previously known algorithms and shows that updates can be
handled in this setting without worsening the delay. Further, our algorithm can also
handle the case of monadic second-order queries with free second-order
variables, for which query answers no longer have constant size: in this case
the delay is linear in the size of each produced answer, as in the static case~\cite{bagan2006mso}.

Our result is shown using the \emph{circuit-based approach} to enumeration
developed 
in~\cite{amarilli2017circuit,amarilli2017circuit_arxiv} and also used
in~\cite{amarilli2018enumeration} to show the result for relabelings.
In this approach, the input is first translated into a circuit representation,
and we impose on the circuit some properties inspired by similar concepts in knowledge compilation,
namely, that it is a complete structured DNNF whose width only depends on the
query.
Then the actual enumeration is
performed on this circuit by relying on these properties.
This approach of translating MSO queries on trees into circuits was
first proposed for probabilistic query evaluation and provenance
computation~\cite{amarilli2015provenance}, and has the advantage of being
applicable to different enumeration problems with very few adaptations, and of
being modular (i.e., in the enumeration we no longer need to refer to the database
or query). 
Our circuit based approach is also related to
factorized~\cite{olteanu2016factorized} representations
where query results are first encoded into a compact representation before
performing other tasks on it, e.g., aggregation.

The circuit-based approach to enumeration lends itself well to relabeling updates,
as we already presented in~\cite{amarilli2018enumeration}: we can toggle
the value of gates when labels change, and then propagate information upwards in
the circuit. The update time then depends on the \emph{depth} of the circuit,
which in our construction is linear in the height of the input tree.
In~\cite{amarilli2018enumeration} we side-stepped this problem by computing a
balanced tree decomposition on the input during the preprocessing phase using~\cite{bodlaender1998parallel}, so
that we applied our construction to a tree of logarithmic height to achieve
logarithmic update time when relabeling nodes of the tree. However,
we did not know how to maintain such a balanced tree decomposition under other
updates like leaf additions and deletions.

Our first contribution is shown by observing that we can combine the techniques
of~\cite{amarilli2018enumeration} with those developed in the recent independent
work~\cite{Niewerth18}. There, it was shown that one can represent trees
as balanced \emph{forest algebra} terms~\cite{BojanczykW07,Bojanczyk10} and
efficiently maintain them under updates. We show that this technique can also be used
in the circuit-based approach, formalizing the forest algebra updates via
\emph{tree hollowings} (in Section~\ref{sec:updates}) and giving an entirely
bottom-up presentation of circuit construction and enumeration (in
Sections~\ref{sec:circuits}--\ref{sec:enumboxes}). This allows us to show our
main result that generalizes the incomparable results
of~\cite{bagan2006mso,kazana2013enumeration}, \cite{NiewerthS18},
\cite{Niewerth18}, and \cite{amarilli2018enumeration}.

The second contribution of this paper is an enumeration
algorithm which is tractable \emph{in combined complexity}. Indeed, the
complexities given in Table~\ref{tab:overview} only apply to data complexity,
i.e., complexity in the database when assuming that the query is fixed. However,
if we want our enumeration algorithms to be applicable, it is important that the
complexity remains tractable, e.g., polynomial, when the query is given as part
of the input. This is an unreasonable hope when the query is
written in MSO, as the combined complexity of Boolean MSO query evaluation (even
without enumeration) is generally nonelementary~\cite{meyer1975weak}. However,
if the query is given as a \emph{tree automaton}, then Boolean evaluation has
tractable combined complexity. But achieving combined tractability for
enumeration is challenging: all results in Table~\ref{tab:overview}
are either intractable in the query or are only tractable when assuming that the
input query is represented as a \emph{deterministic} tree automaton. This
assumption is problematic already in the case of words, as we cannot, e.g.,
convert an input regular expression query on words to a deterministic automaton
in polynomial time. In our recent work~\cite{amarilli2019constant}, we showed
that enumeration \emph{on words} (and without updates) could be performed with
linear-preprocessing and constant delay while being polynomial in an input
\emph{nondeterministic} automaton (presented in the context of information
extraction and \emph{document spanners}~\cite{FaginKRV15}). In the present
paper, we extend these techniques
to make our enumeration algorithm polynomial in
an input nondeterministic tree automaton representation of the query; in other
words we perform enumeration on circuits that are not
\emph{deterministic}~\cite{darwiche2001tractable}. This is performed in
Sections~\ref{sec:nodupes}--\ref{sec:enumboxes} by extending the techniques
of~\cite{amarilli2019constant} to trees. Our results also imply efficient
enumeration algorithms to evaluate the matches of automata on words and
efficiently maintain these results under updates, generalizing some results
of~\cite{amarilli2019constant} to the case of dynamic words.

Our third contribution is to show that our algorithm is close to optimal. We
leverage known lower bounds from data structure research~\cite{AlstrupHR98} to show unconditionally that the update time for an enumeration algorithm for MSO queries has to be $\Omega(\log n/ \log\log n)$ when the delay is constant, or even when it is allowed to be in $o(\log n / \log\log n)$.
Thus, our result is optimal up to a doubly logarithmic factor.

\paragraph*{Paper Structure.}
We give preliminaries in Section~\ref{sec:prelim}. 
We present in Section~\ref{sec:circuits} how to construct from the input tree and
nondeterministic automaton a circuit
representation of the results. We present our enumeration algorithm on such
circuits in Sections~\ref{sec:simple}--\ref{sec:enumboxes}, starting with a
simple scheme in Section~\ref{sec:simple} which enumerates the results with
duplicates and large delays, refining it in
Section~\ref{sec:nodupes} to avoid duplicates, and showing how to achieve the right delay
bound in Section~\ref{sec:enumboxes}. In Section~\ref{sec:updates}, we show
how the algorithm can support updates using~\cite{Niewerth18}, leading to our main results in
Section~\ref{sec:main}. We present our
lower bounds in Section~\ref{sec:lower}, before concluding in
Section~\ref{sec:conc}.
Full proofs are deferred
to the full version~\cite{amarilli2019enumeration_arxiv}.

\section{Preliminaries}
\label{sec:prelim}
We define words, trees, and valuations, present our
automata and a homogenization lemma, and state our problem.

\paragraph*{Trees and Valuations.}
In this paper, we work with trees that are all \emph{rooted} and \emph{ordered}, i.e., there is an order on
the children of each node.
Given a set $\Lambda$ of tree labels, a \mbox{\emph{$\Lambda$-tree}}~$T$
(or \emph{tree} when $\Lambda$ is clear from context)
is a
pair of a rooted tree (also written $T$) and of a \emph{labeling
  function}~$\lambda$
  that maps each node $n$ of~$T$ to a \emph{label}
  $\lambda(n)\in \Lambda$.
  We write $\Leaf(T)$ for the set of leaves of~$T$.
  We abuse notation and identify $T$ with its set of nodes, i.e., we can write
  that $\Leaf(T) \subseteq T$. An \emph{internal node} is a node of $T \setminus
  \Leaf(T)$.
  Until Section~\ref{sec:updates}, all trees in this paper will be \emph{binary},
  i.e., every internal node has
exactly two children, which we refer to as \emph{left} and \emph{right} child.

When evaluating a query with \emph{variables} $\X$ on a
\mbox{$\Lambda$-tree}~$T$,
we will see its possible results as
\emph{valuations}: an \mbox{\emph{$\X$-valuation}} of~$T$
is a function  $\nu:\Leaf(T)\to 2^{\X}$
that assigns to every leaf~$n$ of~$T$ a set of
variables $\nu(n) \subseteq \X$ called the \emph{annotation} of~$n$.
Note that our variables are second-order, i.e., each variable can be
interpreted as a set of nodes of~$T$.
We represent valuations concisely as \emph{assignments}: 
an \emph{$\X$-assignment} is a set $S$ of
\emph{singletons} which are pairs of the form $\langle Z: n\rangle$, where $Z \in
\X$ and
$n \in \Leaf(T)$. The \emph{size} $\card{S}$ of~$S$ is simply the number of
singletons that it contains. There is a clear bijection between $\X$-valuations and
$\X$-assignments, so we write $\alpha(\nu)$ for an $\X$-valuation $\nu$ to
mean the assignment $\{\langle Z: n\rangle \mid Z \in \nu(n)\}$, and
write $\card{\nu} \colonequals \card{\alpha(\nu)}$.
We often 
write $\langle \Y: n\rangle$
for some $n \in T$ and some \emph{non-empty} set 
$\Y \subseteq \X$,
to mean the non-empty set of
singletons $\{\langle Z: n\rangle \mid Z \in \Y\}$.

\paragraph*{Tree Variable Automata}
We will write our query on trees using \emph{tree automata} that
can express, e.g., queries in \emph{monadic second-order logic}
(MSO): see~\cite{thatcher1968generalized} and \cite[Appendix
E.1]{amarilli2017circuit_arxiv}.
Formally, a \emph{tree variable automaton} TVA on binary \mbox{$\Lambda$-trees} for variable
set~$\X$ (or $\Lambda,\X$-TVA) is a tuple $A = (Q, \iota, \delta, F)$,
where $Q$ is a finite set of \emph{states},
$\iota \subseteq \Lambda \times 2^{\X} \times Q$ is the \emph{initial
  relation},
  $\delta \subseteq \Lambda \times Q \times Q \times Q$ is the \emph{transition
  relation}, and $F \subseteq Q$ is the set of
\emph{final states}.
The \emph{size} $\card{A}$ of~$A$ is $\card{Q} + \card{\iota} +
\card{\delta}$. This definition only applies to \emph{binary}
\mbox{$\Lambda$-trees}; the analogous automata for unranked trees will be
introduced in Section~\ref{sec:updates}.

To simplify notation we often see $\delta$ as a tuple of functions, i.e.,
for each $l \in \Lambda$ we have a function $\delta_l: Q\times Q \to 2^Q$ defined
by $\delta_l(q_1, q_2) = \{q \in Q \mid (l, q_1, q_2, q) \in \delta\}$: this
intuitively tells us to which states the automaton can transition on an
internal node with label~$l$ when the states of the two children are
respectively~$q_1$ and~$q_2$. Note that, following our definition of a
valuation and of~$\iota$, the automaton is only reading annotations on leaf nodes.

Having fixed $\Lambda$ and~$\X$,
given a \mbox{$\Lambda$-tree} $T$ and an \mbox{$\X$-valuation} $\nu$ of~$T$,
given a $\Lambda,\X$-TVA $A = (Q, \iota, \delta, F)$,
a \emph{run} of~$A$ on~$T$ under~$\nu$
is a function $\rho: T \to Q$ satisfying the following:
\begin{itemize}
  \item For every $n \in \Leaf(T)$, we have $(\lambda(n), \nu(n), \rho(n))
    \in \iota$;
  \item For every internal node $n$ with label~$l$ and children $n_1, n_2$, 
    we have $\rho(n) \in
    \delta_l(\rho(n_1), \rho(n_2))$.
\end{itemize}
The run is \emph{accepting} if it maps the root of~$T$ to a state in~$F$,
and we say that $A$ \emph{accepts} $T$ under~$\nu$ if there is an accepting
run of~$A$ on~$T$ under~$\nu$.
The \emph{satisfying valuations} of $A$ on $T$ is the set of the
$\X$-valuations $\nu$ of~$T$ such that $A$ accepts~$T$ under~$\nu$, and the
\emph{satisfying assignments} are the corresponding assignments
$\alpha(\nu)$.
Thus, the automaton~$A$ defines a query
on~$\Lambda$-trees with second-order variables $\X$, and its
results on a $\Lambda$-tree~$T$ are the satisfying assignments
of~$A$ on~$T$.

\paragraph*{Homogenization.}
It will be useful to assume a \emph{homogenization} property on automata.
Given a $\Lambda,\X$-TVA $A = (Q, \iota, \delta, F)$, we call $q \in Q$ a
\emph{$0$-state} if there is some $\Lambda$-tree~$T$ and run~$\rho$ of~$A$ on~$T$
that maps the root of~$T$ to~$q$ under the \emph{empty $\X$-valuation} $\nu_\emptyset$
of~$T$ defined as $\nu_\emptyset(n) \colonequals \emptyset$ for each $n \in
\Leaf(T)$. We call $q$ a \emph{$1$-state} if there is some
$\Lambda$-tree~$T$ and run~$\rho$ of~$A$ on~$T$ mapping the root of~$T$ to~$q$
under some \emph{non-empty} $\X$-valuation, i.e., a valuation $\nu$ different
from the empty valuation. Intuitively, a 0-state is a state that~$A$
can reach by reading a tree annotated by the empty valuation, and a 1-state can
be reached by reading a tree with at least one non-empty annotation.
In general a state can be both a 0-state and a 1-state, or it can be neither 
if there is no way to reach it. 
We say that $A$ is \emph{homogenized} if every state is either a $0$-state or
a $1$-state and no state is both a $0$-state and a $1$-state.
We can easily make automata homogenized, by duplicating the states to
remember if we have already seen a non-empty annotation:
\begin{lemmarep}
  \label{lem:homogenize}
  Given a $\Lambda,\X$-TVA $A$, we can compute in linear time a $\Lambda,\X$-TVA
  $A'$ which is homogenized and equivalent to~$A$.
\end{lemmarep}

\begin{proof}
  Let $A = (Q, \iota, \delta, F)$. Intuitively, we build~$A'$ as a product
  of~$A$ with an automaton with two states that remember whether some non-empty
  annotation
  has been seen. 
  Formally, let $A'$ be $(Q', \iota',
  \delta', F')$ where $Q' \colonequals Q \times \{0, 1\}$,
  where $F' \colonequals F \times \{0, 1\}$,
  where $\iota'
  \colonequals \{(l, \emptyset, (q, 0)) \mid (l, \emptyset, q) \in \iota\} \cup
  \{(l, \Y, (q, 1)) \mid (l, \Y, q) \in \iota, \Y \neq \emptyset\}$, and where $\delta' \colonequals
  \{((q_1, i_1), (q_2, i_2), (q, i_1 \lor i_2)) \mid (q_1, q_2, q) \in
  \delta, (i_1, i_2) \in \{0, 1\}^2\}$. We can clearly construct~$A'$ from~$A$ in linear time. 
  We then modify~$A'$ to \emph{trim} it, i.e., removing states that cannot be
  reached by any run, which is clearly doable in linear time by a simple
  reachability test. It is clear from the inductive definition of $0$- and $1$-states
  that $A'$ is homogenized, i.e., each $(q, i)$ for $q \in Q$ and $i \in \{0,
  1\}$ is an $i$-state and it is not an $(1-i)$-state. Further, it is immediate
  by induction that, for any $\Lambda$-tree $T$ and $\X$-valuation $\nu$ of~$T$,
  there is a bijection between the runs of~$A$ on~$T$ under~$\nu$ and the runs
  of~$A'$ on~$T$ under~$\nu$, that maps a run $\rho: T\to Q$ to $\rho': T\to Q'$
  defined by keeping the first component as-is and filling the second component
  on each node~$n$ by~$1$ or~$0$ depending on whether some descendant of~$n$
  has a non-empty annotation or not.
\end{proof}

\paragraph*{Problem statement}
Our goal in this paper is to efficiently enumerate the results of queries on
trees. The inputs to the problem are the $\Lambda$-tree~$T$ and the query given as a
$\Lambda,\X$-TVA~$A$,
with $\Lambda$ the tree alphabet and $\X$ the variables.
The output is the set of the satisfying assignments
of~$A$ and~$T$. We present an 
\emph{enumeration algorithm} to produce them, which first runs a \emph{preprocessing phase} on~$A$
and~$T$: we compute a concise representation of the output as an
\emph{assignment circuit} (Section~\ref{sec:circuits}), and compute an
\emph{index structure} on it (Section~\ref{sec:enumboxes}). Second, 
the \emph{enumeration phase} produces
each result to the query, with no duplicates, while bounding the maximal 
\emph{delay} between two successive answers
(Sections~\ref{sec:simple}--\ref{sec:enumboxes}). 
Third, our algorithm must handle \emph{updates} to~$T$, i.e.,
given an edit operation on~$T$, efficiently update the assignment
circuit and index and restart the enumeration on the updated tree
(Section~\ref{sec:updates}).

Our main result shows how to solve this problem
with preprocessing linear in~$T$ and polynomial in~$A$; with delay
independent from~$T$, polynomial in~$A$, and linear in each
produced assignment; and with update time logarithmic in~$T$ and
polynomial in~$A$. This result is formally stated in Section~\ref{sec:main}.

\section{Building Assignment Circuits}
\label{sec:circuits}
In this section, we start describing the preprocessing phase of our enumeration
algorithm. Given the input TVA and binary tree, we will build an 
\emph{assignment circuit} that concisely represents the set of satisfying
assignments: we do this by adapting the circuit constructions in
our earlier work \cite{amarilli2017circuit,amarilli2018enumeration}
to give them better support for updates and handle non-deterministic automata.

We first define our circuit formalism,
which we call \emph{set circuits}, and their semantics. Second, we define when a set
circuit can serve as an \emph{assignment circuit} for a TVA on a binary tree.
Third, we present properties that our circuits will satisfy, namely, they
are \emph{complete structured DNNFs}, and we can bound a \emph{width} parameter
for them. Last, we state our main circuit construction result at the end of
the section (Lemma~\ref{lem:buildcircuit}): it shows that, given a homogenized
TVA $A$ and binary tree $T$, we can construct an assignment circuit of~$A$
on~$T$ in time $O(\card{A} \times \card{T})$ while respecting our properties and
controlling the width parameter.

The circuit produced by Lemma~\ref{lem:buildcircuit} will then be fed to the
enumeration algorithms presented in
Sections~\ref{sec:simple}--\ref{sec:enumboxes}. To extend enumeration to
updates in Section~\ref{sec:updates}, we will leverage the fact that all our
constructions work by processing~$T$ bottom-up.

\paragraph*{Set Circuits}
A \emph{set circuit} (or just \emph{circuit}) 
$C = (G, W, \mu)$ consists of a directed acyclic graph
$(G, W)$ where $G$ are
the \emph{gates} and $W \subseteq G\times G$ are the \emph{wires}, and of a
function $\mu$ mapping each gate to a type among $\top$, $\bot$, $\var$, $\times$, $\cup$: we will
accordingly call $g$ a $\top$-gate, $\bot$-gate, $\var$-gate, $\times$-gate, or $\cup$-gate depending
on~$\mu(g)$. There is also an injective function $\S_\var$ mapping each
$\var$-gate $g$ to some set
$\S_\var(g)$ of \emph{variables}.
We write $C_\var$ to refer to the variables occurring in~$C$, i.e., 
$C_\var \colonequals \bigcup_{g \in G \mid \mu(g) = \var} \S_\var(g)$.
The set of
\emph{inputs} of a gate~$g \in G$ is $\{g' \in G \mid (g', g) \in W\}$. 
We require that:
\begin{itemize}
  \item $\top$-, $\bot$-, and $\var$-gates have no inputs;
  \item $\times$-gates have exactly $2$ inputs;
  \item $\cup$-gates have at least $1$ input;
  \item $\top$-gates and $\bot$-gates are never used as inputs.
\end{itemize}
The \emph{maximal fan-in} of~$C$ is the maximum number of inputs of a gate
and its \emph{depth} is the maximum length of a directed path of wires.

The goal of set circuits is to concisely represent sets.
Formally,
each gate 
\emph{captures} a 
\emph{set}, in the following sense:
\begin{definition}
  \label{def:capture}
Given a set circuit $C = (G, W, \mu)$, we define for each gate $g \in G$ 
a \emph{captured set}
$\S(g)$ which is a set of subsets of~$C_\var$
defined inductively as follows:
\begin{itemize}
  \item If $g$ is a $\var$-gate then $\S(g) \colonequals \S_\var(g)$.
  \item If $g$ is a $\bot$-gate then $\S(g) \colonequals \emptyset$.
  \item If $g$ is a $\top$-gate then $\S(g) \colonequals
    \{\emptyset\}$.
  \item If $g$ is a $\times$-gate with inputs $g_1$ and $g_2$ then
    $S(g) \colonequals \{S_1 \cup S_2 \mid S_1 \in \S(g_1), S_2 \in \S(g_2)\}$.
  \item If $g$ is a $\cup$-gate then, letting $g_1, \ldots, g_m$ be the inputs
    of~$g$, we define $\S(g) \colonequals \bigcup_{1 \leq i \leq m} \S(g_i)$.
\end{itemize}
\end{definition}

\begin{example}
  \label{exa:circuit}
  Consider the circuit $C$ featuring a $\times$-gate $g$ with one input~$g_1'$
  being a $\var$-gate with $\S_\var(g_1') = \{x\}$, and 
  one input $g_2$ being a $\cup$-gate with two input $\var$-gates $g_2'$ and $g_2''$
  with $\S_\var(g_2') = \{y\}$ and $\S_\var(g_2'') = \{y, z\}$.
  We have $C_\var = \{x, y, z\}$ and $\S(g) = \{\{x, y\}, \{x, y, z\}\}$.
\end{example}

\paragraph*{Assignment Circuits}
An \emph{assignment circuit} for a TVA on a tree is a set circuit where, for each
automaton state~$q$ and tree node~$n$, there is a gate $\gamma(q, n)$ capturing the assignments
for which we reach state~$q$ on node~$n$. Formally:
\begin{definition}
Given a binary $\Lambda$-tree $T$ and a $\Lambda,\X$-TVA $A = (Q, \iota, \delta, F)$, an \emph{assignment circuit}
of~$T$ on~$A$ consists of a circuit~$C$ where $C_\var$ is the set of singletons
$\{\langle Z : n \rangle \mid Z \in \X, n \in \Leaf(T)\}$, and a mapping $\gamma: T \times Q \to C$ such that
for
any $n \in T$ and $q \in Q$,
  the gate $\gamma(n, q)$ is a $\cup$-gate, $\top$-gate, or $\bot$-gate, and
  for any $\X$-valuation~$\nu$ of the subtree $T_n$
of~$T$ rooted at~$n$, 
  we have $\alpha(\nu) \in \S(\gamma(n, q))$ iff there is a run of~$A$ on~$T_n$ under $\nu$
that maps the root node $n$ to state~$q$.
\end{definition}

Note that an assignment circuit concisely represents the satisfying assignments of~$A$
on~$T$: they are $\bigcup_{q \in F} \S(\gamma(n, q))$ where $n$ is the root
of~$T$. Thus, our goal is to use assignment circuits to efficiently enumerate
satisfying assignments. To be able to do this, we will impose certain properties
on assignment circuits, which we now define.

\paragraph*{Complete Structured DNNFs}
The circuits that we build will be \emph{complete structured DNNFs}, and
we will control their \emph{width}. The notion of structured DNNF is inspired by
knowledge compilation~\cite{PipatsrisawatD08}, with DNNF meaning
\emph{decomposable negation normal form}, and completeness is also inspired by
that field~\cite{capelli2019knowledge}; we adapt it here to set circuits
rather than Boolean circuits (as we also did
in~\cite{amarilli2017circuit,amarilli2018enumeration}), with no negations, and
with \emph{decomposability} intuitively implying that no variable can occur in
assignments in both the left and the right input gate of some $\times$-gate.
Formally:

\begin{definition}
  A \emph{v-tree} $\calT$ for a set circuit~$C$ is a binary tree whose leaves
  are labeled by sets of variables that form a partition of~$C_\var$.
  A \emph{structuring function}~$\sigma$ from~$C$ to~$T$ maps each gate $g$ of~$C$ to a
  node $\sigma(g)$ of~$\calT$
  such that:
  \begin{itemize}
    \item For every $\var$-gate $g$, the node $\sigma(g)$ is a leaf of~$\calT$,
      and the variables that label $g$ in~$C$ are a subset of the variables
      of~$\sigma(g)$ in~$\calT$; formally,
      letting $\Y$ be the set of variables that labels~$\sigma(g)$ in~$\calT$, we have
      $\S_\var(g) \subseteq \Y$.
    \item Whenever a gate $g'$ is an input gate to a gate~$g$, then either $g$ and $g'$
      are mapped to the same v-tree node, or the input $g'$ must be a
      $\cup$-gate that is mapped to the
      child of the node of~$g$.
      Formally, for every wire $(g', g) \in W$ of~$C$, either $\sigma(g) =
      \sigma(g')$,
      or $g'$ is a $\cup$-gate and $\sigma(g')$ is a child of~$\sigma(g)$
      in~$\calT$.
    \item Each $\times$-gate $g$ has one input gate $g_1$ such that
      $\sigma(g_1)$ is the left child of~$\sigma(g)$ and one input gate $g_2$
      such that $\sigma(g_2)$ is the right child of~$\sigma(g)$; we call $g_1$
      and $g_2$ the \emph{left} and \emph{right} inputs. Note that, by
      the previous point, 
      $g_1$ and $g_2$ must be $\cup$-gates.
  \end{itemize}
Note that these points
  ensure that $C$ is \emph{decomposable}, namely, for any $\times$-gate $g$ with inputs
$g_1$ and $g_2$, no variable gate in~$C$ has a path both to~$g_1$ and to~$g_2$.
In particular, when defining $\S(g)$ according to
Definition~\ref{def:capture}, there can never be any duplicate in the union that
defines the relational product.

A \emph{complete DNNF structured by~$\calT$} is a set circuit $C$ 
together
  with a v-tree $\calT$ and a 
  structuring function $\sigma$ from~$C$ to~$\calT$.
\end{definition}

\begin{example}
  The circuit $C$ in Example~\ref{exa:circuit} is not a complete structured
  DNNF. However, consider $C'$ built from $C$ where the first input 
  of the 
  $\times$-gate $g$ is now a $\cup$-gate $g_1$ having the $\var$-gate $g_1'$ as
  its only input. Then $C'$ is a complete structured DNNF for the v-tree
  $\calT'$ whose root has a left child $n_1$ labeled $\{x\}$ and a right
  child $n_2$ 
  labeled $\{y, z\}$; the structuring function $\sigma$ maps $g$ to the root
  of~$\calT$, maps $g_1$ and $g_1'$ to~$n_1$, and $g_2$, $g_2'$, and $g_2''$
  to~$n_2$.
\end{example}

When we have a complete structured DNNF $C$ with a structuring function $\sigma$
to a v-tree $\calT$, we see the gates of~$C$ as partitioned into \emph{boxes},
with each box being the preimage of some node of~$\calT$ by~$\sigma$. 
We use $\boxf(g)$ to denote the box of some gate $g$
(formally, $\boxf(g)
\colonequals \sigma^{-1}(\sigma(g))$),
We talk
about the \emph{tree of boxes} to mean the structure on boxes that
follows~$\calT$.
In particular, given a box $B$, letting $n$ be the node of~$\calT$ such that $B =
  \sigma^{-1}(n)$, if $n$ is an internal node then we call $B$ a \emph{non-leaf box}
  and denote by $\leftb(B)$ and $\rightb(B)$ its \emph{left child
  box} and \emph{right child box} in the tree of boxes, i.e., $\sigma^{-1}(n_1)$ and
  $\sigma^{-1}(n_2)$ respectively, where $n_1$ and $n_2$ are the children of~$n$
  in~$\calT$.
We will use boxes to
define a structural parameter of complete structured DNNFs, called \emph{width},
and similar to width in~\cite{capelli2019knowledge}.

\begin{definition}
The \emph{width} of a structured complete DNNF is the maximal
  number of $\cup$-gates in a box, i.e., \[\max_B \card{\{g \in B \mid
  \mu(g) = \cup\}}\;.\]
\end{definition}
While this notion of width only bounds the number of $\cup$-gates in each box, we can always rewrite a
structured complete DNNF of width~$w$ in linear time to ensure that the number of
$\times$-gates in each box is also bounded by 
$w^2$~\cite[Observation~3]{capelli2019knowledge}. Intuitively, each $\times$-gate
in a box~$B$ has two $\cup$-gates as input, one in $\leftb(B)$ and one in
$\rightb(B)$, so there are at most $w^2$ non-equivalent combinations.
Hence, we will assume that this bound holds on all circuits that we manipulate
(and in fact our circuit construction obeys it directly, with no rewriting 
needed).

\paragraph*{Building Assignment Circuits}
We have defined the assignment circuits that we want to compute, defined the
notion of a structured complete DNNF and a width parameter for them.
We can now state our main result for this section, namely,
that we can efficiently construct assignment circuits. Observe that, while the
depth of the circuit depends on the input tree, the width only
depends on~$\card{Q}$, which will be crucial for our delay bounds.

\begin{lemmarep}\label{lem:buildcircuit}
  Given any binary $\Lambda$-tree $T$
  and homogenized \mbox{$\Lambda,\X$-TVA}
  $A=(Q,\iota,\delta,F)$, we can construct in time $O(\card{T} \times \card{A})$
  a structured complete DNNF $C$ which is an assignment
  circuit of~$A$ and~$T$,
  a v-tree~$\calT$,
  and a structuring function
  from~$C$ to~$\calT$, such that $C$ has 
  width~$\card{Q}$
  and depth $O(\height(T))$.
\end{lemmarep}
\begin{proofsketch}
  We construct $\calT$ by taking $T$, removing all node labels, and labeling
  each leaf node $n$ by the set of singletons $\langle \X: n\rangle$:
  thus, $\calT$ is a v-tree for the set of variables
  $C_\var = \{\langle Z: n \rangle \mid Z \in \X, n \in T\}$ of~$C$
  given by the definition of assignment circuits.

  We now build~$C$ bottom up.
  For a leaf node $n$ of~$T$ with label $l \in
  \Lambda$, we build the box $B_n$ for~$n$ by setting $\gamma(n, q)$ for all $q
  \in Q$ as:
  \begin{itemize}
  \item a $\bot$-gate, if there are no tuples of the form  $(l,\Y,q)  \in \iota$,
  \item a $\top$-gate, if $(l, \emptyset, q) \in \iota$, and
  \item a $\cup$-gate having as inputs one variable gate labeled by
    $\langle \Y:n\rangle$ for each non-empty $\Y \subseteq \X$ such that
    $(l,\Y,q) \in \iota$, otherwise.
  \end{itemize}
  As $A$ is homogenized, the first and last case are disjoint, i.e., we cannot
  have both $(l, \emptyset, q) \in \iota$ and $(l, \Y, q) \in \iota$ with 
  $\Y \neq \emptyset$.

  For an inner node $n$ of $T$ with label $l$ and child nodes $n_1$
  and $n_2$, we construct the box $B_n$ as follows.
  For every triple $(q_1,q_2,q) \in \delta_l$, we define a
  $\times$-gate $g^{q_1,q_2}$ with inputs $\gamma(n_1,q_1)$ and
  $\gamma(n_2,q_2)$. If there is no such triple, we let $\gamma(n,q)$
  be a $\bot$-gate. Otherwise, we let $\gamma(n,q)$ be a $\cup$-gate
  that has all such $\times$-gates $g^{q_1,q_2}$ as input.

  In terms of accounting, it is clear that 
  there are at most $|Q|$ \mbox{$\cup$-gates}
  in each $B_n$, that the depth of the circuit is as stated, and the construction of the whole circuit is in time
  $O(\card{A} \times \card{T})$ as promised.
  Now, a straightforward bottom-up induction on~$T$ shows that the gates
  $\gamma(n, q)$ capture the correct set for any~$n$, i.e., that for any leaf
  node $n$ and any $q \in Q$ we have: 
  $
    \S(\gamma(n, q)) = \{\langle \Y: n\rangle \mid
    (\lambda(n), \Y, q) \in \iota\}
  $
  and for any internal node $n$ with label~$l$ and children $n_1$ and $n_2$ and any $q \in Q$
  we clearly have:
  \[
    \S(\gamma(n, q)) \;\;= \bigcup_{(q_1, q_2, q) \in \delta_l} \S(\gamma(n_1,
    q_1)) \times \S(\gamma(n_2, q_2))
  \]

  It is easy to check that all rules of assignment circuits
  are respected with the exception of the rule that $\top$- and
  $\bot$-gates are never allowed as inputs to other gates. In
  particular, all $\cup$-gates and $\times$-gates have the right
  fan-in.
  To ensure that $\top$- and $\bot$-gates are never used as inputs,
  one can use a slightly modified construction (see appendix) that avoids adding
  outgoing wires to $\top$- and $\bot$-gates by treating these cases
  in a special way. This uses the fact that, as $A$ is homogenized,
  there is no gate $\gamma(n,q)$ that captures both the empty
  assignment and some non-empty assignment.
\end{proofsketch}
\begin{proof}
  We construct $\calT$ by taking $T$, removing all node labels, and labeling
  each leaf node $n$ by the set of singletons $\langle \X: n\rangle$:
  note that $\calT$ is indeed a v-tree for the set of variables
  $C_\var = \{\langle Z: n \rangle \mid Z \in \X, n \in T\}$ of~$C$
  given by the definition of assignment circuits.

  We now present the construction of~$C$ bottom up.
  We first describe the case of a leaf node $n$ of~$T$ with label $l \in
  \Lambda$. In this case, 
  we construct the box $B_n$ for~$n$ as follows:
  \begin{itemize}
    \item For every $0$-state $q$ of~$A$, we set $\gamma(n, q)$ to be a $\top$-gate
      if $(l, \emptyset, q) \in \iota$, and a $\bot$-gate otherwise.
    \item For every $1$-state $q$ of~$A$ with no tuples of the form $(l, \Y, q)$
      in~$\iota$, we set $\gamma(n, q)$ to be a $\bot$-gate.
    \item For every $1$-state $q$ of~$A$ with at least one tuple of
      the form $(l, \Y, q)$, we set $\gamma(n, q)$ to be a $\cup$-gate
      having as inputs one variable gate labeled by $\langle \Y:n\rangle$ for
      each $\Y \subseteq \X$ such that $(l,\Y,q) \in \iota$.
      Note that $\Y$ is then nonempty because $q$ is a 1-state.
  \end{itemize}
  It is clear that 
  $B_n$ has at most $|Q|$ $\cup$-gates
  and that
  all restrictions for structured complete DNNFs are met.

  For an inner node $n$ of $T$ with label $l$ and child nodes $n_1$
  and $n_2$, we construct the box $B_n$ as follows. First, for every $0$-state
  of $A$, we set $\gamma(n,q)$ to be a $\top$-gate if and only if
  there are states $q_1$ and $q_2$ in $A$ such that
  $(q_1,q_2,q) \in \delta_l$ and $\gamma(n_1,q_1)$ and
  $\gamma(n_2,q_2)$ are both $\top$-gates. Otherwise, we set
  $\gamma(n,q)$ to be a $\bot$-gate.

  Second, for every $1$-state $q$ of $A$ and every triple
  $(q_1,q_2,q) \in \delta_l$, let $g_1\colonequals \gamma(n_1,q_1)$
  and $g_2 \colonequals \gamma(n_2,q_2)$.
  
  We define a gate $g^{q_1,q_2}$ such that we have the
  equality:
  \[\S(g^{q_1,q_2}) = \S(g_1) \times \S(g_2) \tag{*}\]
  but
  while respecting the rule that $\top$ and $\bot$-gates can never be
  used as input to another gate.
  Specifically:
  \begin{itemize}
    \item If one of~$g_1$, $g_2$ is a $\bot$-gate, we set $g^{q_1,q_2}$ to be
      a $\bot$-gate, which clearly satisfies (*);
    \item If one of~$g_1$, $g_2$ is a $\top$-gate, we set $g^{q_1,q_2}$ to be
      the other gate; this also satisfies (*);
     \item Otherwise we set $g^{q_1,q_2}$ to be a $\times$ gate with inputs $g_1$ and $g_2$.
  \end{itemize}
  Having created the necessary gates $g^{q_1,q_2}$ for the triples
  of~$\delta_l$, 
  we now create $\gamma(n, q)$ for every $1$-state~$q$
  as a gate that satisfies:
  \[\S(\gamma(n, q)) = \bigcup_{(q_1, q_2, q) \in \delta_l} \S(g^{q_1, q_2}) \tag{**}\]
  Specifically:
  \begin{itemize}
    \item If all $g^{q_1,q_2}$ in the union are $\bot$-gates (in
      particular if the union is empty), we set
      $\gamma(n,q)$ to also be a $\bot$-gate, respecting (**);
    \item Otherwise we exclude all $\bot$-gates  from the union and 
      set $\gamma(n,q)$ to be a $\cup$-gate, which has all remaining gates
      $g^{q_1,q_2}$ as input, satisfying (**).
  \end{itemize}
  We
  can easily check that all rules of assignment circuits are respected. In
  particular,
  all $\cup$-gates and $\times$-gates have the right fan-in.
  To check that we never use $\top$ and $\bot$ as input to another
  gate, the only subtlety is that, when defining the $\cup$-gate $\gamma(n, q)$
  for a 1-state~$q$, we must check that $g^{q_1,q_2}$ can never be a
  $\top$-gate, but this is because one of $q_1$ and $q_2$ must be a 1-state,
  hence it cannot be a 0-state because $A$ is homogenized; now it can be seen by
  induction that whenever $\gamma(n', q')$ is a $\top$-gate then $q'$ is a
  0-state. It is also clear that the definition of a structured complete DNNF is
  respected, in particular the inputs to
  $\times$-gates are $\cup$-gates in the two child boxes.

  In terms of accounting, it is clear that 
  there are at most $|Q|$ $\cup$-gates
  in each $B_n$, that the depth of the circuit is as stated, and the construction of the whole circuit is in time
  $O(\card{A} \times \card{T})$ as promised.
  Last, a straightforward bottom-up induction on~$T$ shows that the gates
  $\gamma(n, q)$ capture the correct set for any~$n$, i.e., that for any leaf
  node $n$ and any $q \in Q$ we have:
  \[
    \S(\gamma(n, q)) \;\;= \;\;\{\langle \Y: n\rangle \mid
    (\lambda(n), \Y, q) \in \iota\}
  \]
  and for any internal node $n$ with label~$l$ and children $n_1$ and $n_2$ and any $q \in Q$
  we clearly have the following, by (*) and (**) (and their analogues in the case
  of 0-states):
  \[
    \S(\gamma(n, q)) \;\;= \bigcup_{(q_1, q_2, q) \in \delta_l} \S(\gamma(n_1,
    q_1)) \times \S(\gamma(n_2, q_2))
  \]
  Hence, the construction is correct, which concludes the proof.
\end{proof}

\section{Simple Enumeration Algorithm}
\label{sec:simple}
In the three following sections, we will present how to enumerate the set
of assignments captured by gates of assignment circuits. 
We start in this section by presenting an algorithm which is simple 
but has two important drawbacks. First, the worst-case delay is
$O(\depth(C))$, i.e., linear in the depth of the circuit $C$.  Second,
assignments are output multiple times.
We will refine this algorithm in Section~\ref{sec:nodupes} 
to ensure that every assignment is enumerated exactly once.
Last in Section~\ref{sec:enumboxes}, we show how to
bound the delay by the width of the circuit (instead of the depth).

To define our enumeration algorithms in this and the following sections, we
introduce some useful notation. 
For any $\cup$-gate $g$ of~$C$,
for any gate $g'$ of~$C$, 
we write
$g' \cuppath g$ if there is a path $g' = g_1, \ldots, g_n = g$ from~$g'$ to~$g$
in~$C$ where each $(g_i, g_{i+1})$ is a wire in~$W$ and where all intermediate gates
$g_2, \ldots, g_n$ are $\cup$-gates. We then write $\transitive{g}$ to mean the
set of $\var$-gates and $\times$-gates $g'$ such that $g' \cuppath g$.
The following observation shows why $\transitive{g}$ is useful for enumeration:
it is proven by an immediate induction on~$\cup$-gates:

\begin{observation}
  \label{obs:transitiveq}
  For any $\cup$-gate $g \in C$, we have that
   $\S(g) = \bigcup_{g' \in \transitive{g}}
  \S(g')$.
\end{observation}

  We observe that we can enumerate $\transitive{g}$ for every
  gate $g$ 
  by doing a simple preorder traversal
  of the circuit: however, doing this naively only ensures a delay of
  $O(\depth(C))$, and it enumerates each gate $g'$ as many times
  as there are paths that witness $g' \cuppath g$ in~$C$. 
  We denote by \textsc{enum}$_{\downarrow}^{\text{\textsc{dupes}}}(g)$ this naive
  procedure.

Using this procedure, we present our enumeration algorithm for
$\S(g)$ as Algorithm~\ref{alg:enumsimple}.
The algorithm applies to any decomposable set 
circuit and does not use the v-tree or the structuring function.

\begin{algorithm}[tb]
  \caption{Simple enumeration algorithm}\label{alg:enumsimple}
  \begin{algorithmic}[1]
    \Procedure{enum$_{\S}^{\text{dupes}}$}{$g$}
      \For{$g' \in \Call{enum$_{\downarrow}^{\text{\textsc{dupes}}}$}{g}$}
        \If{$g'$ is a $\var$-gate}
          \textbf{output} $\{\S_\var(g')\}$
        \Else \Comment{$g'$ is a $\times$-gate}
          \For{$S_L \in \Call{enum$_{S}^{\text{dupes}}$}{\text{left input of $g'$}}$} \label{line:enuml}
            \For{$S_R \in \Call{enum$_{S}^{\text{dupes}}$}{\text{right input of $g'$}}$} \label{line:enumr}
              \State \textbf{output} $S_L \cup S_R$
            \EndFor
          \EndFor
        \EndIf
      \EndFor  
    \EndProcedure
  \end{algorithmic}
\end{algorithm}

  Algorithm~\ref{alg:enumsimple} is presented using ``output'' statements to produce new
  results (like, e.g., Python's ``yield''). When we recursively use
  the enumeration algorithm on a subcircuit (as in
  lines~\ref{line:enuml} and~\ref{line:enumr}), we assume that this
  enumeration is started in another thread, which will run until the
  first output is produced. Afterwards the new thread pauses until the calling
  thread requests the next value. Whenever the calling thread requests
  a new value, the called thread runs until it produces the next
  output.
The following result is now not hard to see.  
\begin{proposition}
  \label{prop:enumsimple}
  Given a structured complete DNNF $C$ and
  \mbox{$\cup$-gate}~$g$, Algorithm~\ref{alg:enumsimple} enumerates
  $\S(g)$ (with duplicates) with delay $O(\depth(C) \times |S|)$, where $S$ is the produced
  assignment.
\end{proposition}

\begin{proof}
  It is clear that Algorithm~\ref{alg:enumsimple} is correct, because it
  directly follows Observation~\ref{obs:transitiveq}, and the inputs to
  $\times$-gates are always \mbox{$\cup$-gates} so we always call the algorithm on a
  $\cup$-gate.
  In terms of delay, each assignment of size $k$ produced by the algorithm
  required at most $2k-1$
  recursive calls, each of which correspond to a $\var$-gate or  $\times$-gate used when
  producing the assignment: note that this uses the decomposability of~$C$,
  and uses the fact that $\top$-gates are never used as inputs to another
  gate.
  Now, each recursive call has delay $O(\depth(C))$ 
  for the call to \textsc{enum}$_{\downarrow,\text{\textsc{dupes}}}(g)$, hence the delay of the algorithm is as
  claimed, which completes the proof.
\end{proof}

As a side remark, note that the number of times that
Algorithm~\ref{alg:enumsimple} enumerates each assignment 
is related to the number of runs of the TVA for this assignment.
Specifically, up to redefining Definition~\ref{def:capture} with multisets,
and up to small changes in Lemma~\ref{lem:buildcircuit}, we
could ensure that each assignment in $\S(\gamma(n, q))$ is enumerated exactly as
many times as there are runs on the subtree
rooted at~$n$ under the corresponding valuation such that the root node is mapped to~$q$.

\section{Eliminating Duplicates}
\label{sec:nodupes}
In this section, we adapt Algorithm~\ref{alg:enumsimple} to enumerate
satisfying assignments without duplicates.
A simple
idea would be to change Algorithm~\ref{alg:enumsimple} to enumerate
the gates of 
$\transitive{g}$ without duplicates.
Sadly, this 
would not suffice:
imagine that we enumerate $\S(g)$ for some $\cup$-gate $g$ having two inputs
$g_1$ and~$g_2$ for which
$\S(g_1) \cap \S(g_2) \neq \emptyset$,
then if we consider $g_1$ and~$g_2$ separately we will enumerate their common
assignments twice.
However, the crucial point is that $\S(g_1) \cap \S(g_2) \neq \emptyset$ implies
that 
$g_1$ and $g_2$ are in the same box, thanks to the
following property of structured complete DNNFs:

\begin{lemmarep}\label{lem:lca1}
  For any structured complete DNNF $C$,
  for any $\var$-gate or $\times$-gate $g$ of~$C$ and assignment $S$,
  if we have $S \in \S(g)$,
  then the box of $g$ is the (unique) least common ancestor of the boxes that
  contain the $\var$-gates whose variables occur in $S$.
\end{lemmarep}
\begin{proofsketch}
  This is because $\top$-gates are not
  allowed as input to any gate and $\times$-gates always use inputs
  from both subtrees.
\end{proofsketch}
\begin{proof}
  The definition of a structured complete DNNF clearly ensures that, letting $B
  \colonequals \boxf(g)$, for any var-gate $g'$ whose variables occur in $S$, the leaf box $B'$ that contains $g'$
  must be a descendant of~$B$ in the tree of boxes. Hence, $B$ must be a common ancestor of the boxes that contain the $\var$-gates with
  the variables of~$S$. Now, if $B$ is a leaf box then it is necessarily the lowest common
  ancestor. Otherwise, $g$ must be a $\times$-gate, so it has two inputs $g_1$ and~$g_2$
  which must be in the two child boxes $B_1$ and $B_2$, and there are two
  assignments $S_1 \in \S(g_1)$ and $S_2 \in \S(g_2)$ such that $S_1 \cup
  S_2 = S$; further, as $\top$-gates are never used as input to another gate,
  neither $S_1$ nor $S_2$ are empty. Hence, $B_1$ and $B_2$ are both ancestors
  of some of the boxes that contain the $\var$-gates whose variables occur in~$S$, which clearly
  implies that their parent $B$ cannot be a strict ancestor of the lowest common
  ancestor of this set of boxes.
\end{proof}

This observation leads to the idea of \emph{boxwise enumeration}, i.e., 
simultaneously considering a set of gates that are all in the same box, and
enumerate \emph{simultaneously} the assignments that they capture, without duplicates. This idea
is reminiscent of evaluating a nondeterministic automaton on a word by
determinizing the automaton on-the-fly, and it was already used
in~\cite{amarilli2018enumeration} in the case of words; we will extend it to
trees. We will accordingly call \emph{boxed set} a set~$\Gamma$ of gates that are
all $\cup$-gates and that all belong in the same box. We write $B_\cup$ for a
box~$B$ to mean the $\cup$-gates of~$B$.

Given a boxed set $\Gamma$ in some box~$B$,
let us denote by $\S(\Gamma)$ the set of assignments
$\bigcup_{g \in \Gamma} \S(g)$, which we want to enumerate without duplicates,
and let us write $\transitive{\Gamma} \colonequals \bigcup_{g\in
\Gamma}\transitive{g}$.
Our enumeration will rely on a procedure $\Call{box-enum}{\Gamma}$ that
enumerates the boxes $B'$ such that $\transitive{\Gamma} \cap B' \neq
\emptyset$ 
and produces for each such box~$B'$ the
\emph{$\cup$-reachability relation} between $B'$ and~$\Gamma$, i.e., 
the binary relation $R(B', \Gamma)$ describing which gates
of~$B'_\cup$ have a
path of $\cup$-gates to~$\Gamma$. 
Formally, we define the $\cup$-reachability
relation between any sets of gates~$G'$ and~$G$ as
$R(G', G) \colonequals \{(g', g) \in G'_\cup \times G_\cup  \mid g' \cuppath
g\}$.
Pay attention to the fact that
each call to \textsc{box-enum} returns the complete relation $R(B', B)$ for one
of the boxes~$B'$ (i.e., we do not enumerate the pairs of~$R(B', B)$), and the
relation for each box~$B'$ is returned only once (i.e., there should not be
duplicate boxes).

It is straightforward to show that
\Call{box-enum}{$\Gamma$} can be implemented with delay
$O(\depth(C)\times\poly(w))$, where $w$ is the width of the circuit, by
exploring the boxes from~$B$
(traversing only $\cup$-gates and at most one other gate per level)
and
maintaining the information $R(B', \Gamma)$ for the boxes~$B'$ that we visit.
In the next section, we show how we can implement \textsc{box-enum} more
efficiently.

The point of \textsc{box-enum} is the following
easy consequence of Observation~\ref{obs:transitiveq}:

\begin{observation}
  \label{obs:transitiveq2}
  For any boxed set $\Gamma$ in any box, we have:
\[
  \S(\Gamma) \;\;= \biguplus_{R(B',\Gamma) \in \Call{box-enum}{\Gamma}}
  \bigcup_{\substack{
    g' \in W \circ R(B',\Gamma)
    \\\text{with~} \mu(g') \in \{\var, \times\}}}
  \S(g')
\]
where $W$ is the set of wires of~$C$, where $W \circ R(B',\Gamma)$ 
denotes the composition of the two binary relations, and where the outermost
  union is without duplicates.
\end{observation}

\begin{proof}
To see why the equality holds,
first observe that for each $R(B', \Gamma)$, the inner union goes over all gates
  of $\transitive{\Gamma} \cap B'$. Indeed, the definition of the $\cup$-reachability
  relation ensures that $R(B', \Gamma)$ denotes the $\cup$-gates $g'$
  of~$B'$ such that $g' \cuppath g$ for some $g\in\Gamma$, so the inner union goes over all their
  inputs that are $\times$-gates and $\var$-gates (and which must also be in
  box~$B'$). 
  Now, by definition of \textsc{box-enum}, the outermost union goes over all gates of
  $\transitive{\Gamma}$. Thus, we can conclude thanks to
  Observation~\ref{obs:transitiveq}. The fact that the outermost union is
  disjoint follows immediately from
  Lemma~\ref{lem:lca1}.
\end{proof}

Observation~\ref{obs:transitiveq2} suggests that we can perform enumeration
without duplicates recursively, simply
by re-applying the scheme on the inputs of the gates of~$\transitive{\Gamma}
\cap B'$ to perform enumeration. The details are
subtle, however, and this is why we designed \textsc{box-enum} to return more than
just the set
$\pi_1(R(B',\Gamma))$ of the $\cup$-gates of~$B'$ having a path of
$\cup$-gates to~$\Gamma$ (with $\pi_1$ denoting projection to the first
component):
we will really need the complete relation $R(B', \Gamma)$ for the recursive
calls, in order to avoid duplicate assignments across
multiple $\times$-gates in $\transitive{\Gamma} \cap B'$.

\begin{algorithm}[tb]
  \caption{Enumeration algorithm without duplicates}\label{alg:enum}
  \begin{algorithmic}[1]
    \Procedure{enum$_{\S}$}{$\Gamma$}
      \State $B \gets \boxf(\Gamma)$
      \For{$R(B',\Gamma) \in \Call{box-enum}{\Gamma}$}
        \State $G' \gets \pi_1(W \!\circ R(B', \Gamma))$ \Comment{project to first
        component}
        \State $G_\var \gets \{ h \in G' \mid \mu(h) = \var\}$
        \For{$g'' \in G_\var$}
          \State \textbf{output} $(\S_\var(g''),\{g''\} \circ W \circ R(B', \Gamma))$\label{alg:enum:output1}
        \EndFor
        \State $G_\times \gets \{ h \in G' \mid \mu(h) = \times\}$ \label{alg:enum:recstart}
        \State $\Gamma_L \gets $ set of left inputs of $G_\times$  
        \For{$(S_L,\Gamma_L') \in \Call{enum$_S$}{\Gamma_L}$}
          \State $G_\times' \gets $ set of gates of~$G_\times$ with left input 
          in~$\Gamma_L'$
          \State $\Gamma_R \gets $ set of right inputs of~$G_\times'$
          \For{$(S_R,\Gamma_R') \in \Call{enum$_S$}{\Gamma_R}$}
            \State $G_\times'' \gets $ set of gates of~$G_\times'$\! with right input 
            in $\Gamma_R'$
            \State $\Gamma' \gets G_\times'' \circ W \circ
            R(B',\Gamma)$\label{alg:enum:circ}
            \State \textbf{output} $(S_L \cup S_R, \Gamma')$\label{alg:enum:output2}
          \EndFor
        \EndFor
      \EndFor  
    \EndProcedure
  \end{algorithmic}
\end{algorithm}

Our algorithm to enumerate~$\S(\Gamma)$ is presented as Algorithm~\ref{alg:enum}.
The semantics are changed a bit relative to Algorithm~\ref{alg:enumsimple}.
Algorithm~\ref{alg:enum} takes as input a boxed set
$\Gamma$, and the output is the
enumeration of $\S(\Gamma)$ without duplicates.
Moreover, for each
assignment $S$ in this set, the algorithm also returns its \emph{provenance
relative to~$\Gamma$}, i.e., the subset $\Prov(S, \Gamma) \colonequals \{g \in
\Gamma \mid S \in \S(g)\}$, which again is used for the recursive calls.

\begin{theoremrep}
  \label{thm:enumnodupes}
  Given a structured complete DNNF $C$ and given a boxed set $\Gamma$,
  we can enumerate $\S(\Gamma)$ (without duplicates)
  with delay $O(|S| \times (\Delta + w^3))$, where $S$ is the produced
  assignment, $\Delta$
  is the delay of \textsc{box-enum}, and $w$ is the width of~$C$. Further,
  we correctly produce for each assignment $S$ its provenance $\Prov(S, \Gamma)$
  relative to~$\Gamma$. 
\end{theoremrep}
\begin{proofsketch}
  We have to show three things to establish correctness: (1) For every output
  $(S,\Gamma')$ of the
  algorithm, we have $\Gamma' \subseteq \Gamma$ and 
  $S \in \S(g)$ for every $g \in \Gamma'$. (2) For
  every assignment $S \in \S(g)$ with $g \in \Gamma$, we have some output
  $(S,\Gamma')$ with $g \in \Gamma'$. (3) No assignment $S$ is
  outputted twice.

  Statements (1) and (2) can be shown by induction on the number of
  variable gates that contribute to the assignment. For the induction
  base case, one can verify that the assignments that only use one variable
  gate are correctly handled in Line~\ref{alg:enum:output1}.
  For the induction step, one can verify that the topmost
  $\times$-gate $g_\times$ that is involved in computing an assignment
  is handled correctly in the lines~\ref{alg:enum:recstart}
  to~\ref{alg:enum:output2}.

  Statement (3) follows from Lemma~\ref{lem:lca1} and the fact
  that for each box $\textsc{box-enum}(\Gamma)$ returns at most one
  relation.

  The proof for the runtime delay is similar to
  that of Proposition~\ref{prop:enumsimple}, with the difference that the local
  computations are more expensive as we need to compute relational
  compositions, e.g., $W \circ R(B', \Gamma)$.
\end{proofsketch}
\begin{proof}
  We first show (1). The proof is by induction over the number of $\var$-gates
  whose variables occur in the assignment.
  The base case is that of the assignments produced at Line~\ref{alg:enum:output1}, whose assignments
  are clearly correct.

  For the induction case, we consider assignments produced at
  Line~\ref{alg:enum:output2}. By the induction hypotheses
  $S_L \in \S(g_L)$ for every $g_L \in \Gamma_L'$ and $S_R \in \S(g_R)$ for every
  $g_R \in \Gamma_R'$. Note now that the gates of $G_\times''$ have their left input
  in~$\Gamma_L'$ and their right input in~$\Gamma_R'$, so indeed we have $S \in
  \S(g'')$ for every $g'' \in G_\times''$. Now, for every $g'' \in 
  G''_\times \circ W \circ R(B', \Gamma)$, we know that there is a
  $\cup$-gate $g'$ and a gate $g\in\Gamma$
  such that $(g'', g') \in W$ and $(g', g) \in R(B', \Gamma)$, i.e., by
  definition, $g''
  \cuppath g$, so this witnesses that $S \in \S(g)$ for some $g\in\Gamma$, concluding the proof of
  (1).

  Now we show (2). The proof is again by induction over the size of
  assignments. The base case is again Line~\ref{alg:enum:output1}, which
  correctly produces all assignments that only involve one variable gate,
  by definition of \Call{box-enum}{$\Gamma$}.

  Let now $S \in \S(\Gamma)$ be some assignment involving multiple $\var$-gates.
  Clearly there must be a $\times$-gate $g_\times$ in $\transitive{\Gamma}$
  witnessing that $S \in \S(\Gamma)$, i.e., such that we have
  $S \in \S(g_\times)$. Let $g_L$ and $g_R$ be the left and right inputs
  of~$g_\times$, and let $S_L' \in \S(g_L)$ and
  $S_R' \in \S(g_R)$ be the assignments witnessing that $S \in
  \S(g_\times)$, i.e., we have $S = S_L' \cup S_R'$.

  Let $B'$ be the box of~$g_\times$. As $g_\times \in \transitive{\Gamma}$, we
  know that \Call{box-enum}{$\Gamma$}
  returns some relation $R(B', \Gamma)$, and then we know that 
  $G_\times$ contains $g_\times$.
  Therefore, we can conclude that $g_L \in \Gamma_L$. By the
  induction hypothesis we know that
  $(S_L',\Gamma_L') \in \textsc{enum}_{\S}(\Gamma_L)$ and that $g_L \in
  \Gamma_L'$. Therefore, we can conclude that $g_\times \in G_\times'$, so that
  $g_R \in \Gamma_R$ and, again
  using the induction hypothesis, we have
  $(S_R',\Gamma_R') \in \textsc{enum}_{\S}(\Gamma_R)$ with $g_R \in \Gamma_R'$. It
  follows that $g_\times \in G_\times''$.
  Now, as $R(B',\Gamma)$ was correctly computed by \Call{box-enum}{$\Gamma$}, we have
  $(g_\times, g) \in R(B',\Gamma)$, so that $g \in G_\times'' \circ W \circ
  R(B', \Gamma)$. Hence, we
  indeed produce $(S, \Gamma')$ with a set $\Gamma'$ that contains~$g$.

  At last, we show (3), i.e., that no assignment is output more than
  once. First observe that, by 
  Lemma~\ref{lem:lca1}, the assignments captured by the gates of
  different boxes are disjoint, so that it suffices to show the claim for each
  $R(B',\Gamma)$. For assignments involving only one $\var$-gate, the claim is immediate as the
  assignments that are produced all involve a different variable gate, and the
  labeling function $S_\var$ of variable gates is injective.
  For assignments involving multiple $\var$-gates, we use the fact that by induction the recursive
  calls on $\Gamma_L$ and $\Gamma_R$ output each assignment once, and the
  properties of a structured complete DNNF ensures that each assignment $S \in
  \S(g)$ for a $\times$-gate $g$ has a unique partition (given following the
  v-tree) as $S_L \cup S_R$
  with $S_L \in \S(g_L)$ and $S_R \in \S(g_R)$ for $g_L$ and $g_R$ the left and
  right inputs of~$g$, respectively.

  For delay, the analysis is similar to that of
  Proposition~\ref{prop:enumsimple}: producing an assignment $S$ requires again
  $2\card{S}-1$ recursive calls, and the delay of each call includes the delay
  $\Delta$ of the call to \textsc{box-enum}, plus the delay of the operations
  performed in Algorithm~\ref{alg:enum} which are bounded by~$O(w^3)$ for~$w$
  the width of the circuit. Indeed, remember in particular that each box
  contains at most $w^2$ $\times$-gates; the number of $\var$-gates in each leaf
  box is unbounded but we make progress each time we examine one such gate.
\end{proof}

\section{Enumerating Boxes Efficiently}
\label{sec:enumboxes}
We have shown in the previous section how to enumerate the assignments captured
by an assignment circuit without duplicates.
However the delay depends on the delay of \textsc{box-enum}, which for the naive 
implementation we discussed before was linear in the depth of the circuit. This 
leads to a delay for the assignment enumeration that is also linear in the 
depth of the circuit.
In this section, we show how to speed up
\textsc{box-enum} and make its delay independent from the depth of the circuit,
using a similar idea to \emph{jump pointers}
from~\cite{amarilli2019constant}. As the delay added by
Algorithm~\ref{alg:enum} only depends on the circuit width,
this will establish our
overall delay bound, which we state at the end of the section.

\paragraph*{Interesting and Bidirectional Boxes}
To speed up \textsc{box-enum}, we will need to perform some linear-time
preprocessing on the input circuit, following the tree of boxes. Let us first
give the required definitions.
For each boxed set~$\Gamma$, the set of \emph{interesting boxes} for~$\Gamma$ is defined as
$\{ B' \mid B' \cap \transitive{\Gamma} \neq \emptyset \}$: these are the boxes
that
\Call{box-enum}{$\Gamma$} must consider. We also define
the set of \emph{bidirectional boxes} for~$\Gamma$ as
\[\{B' \mid \transitive{\Gamma} \text{ intersects boxes in both subtrees of } B'\}\]
These boxes are necessarily non-leaf boxes and have
  interesting boxes as descendants of their left and right child.
Note that a bidirectional box may also be interesting.

One key idea to optimize \textsc{box-enum} is to ``jump'' from a box $B$ to a
bidirectional descendant box $B'$, skipping boxes on the path from $B$
to~$B'$ that
are neither interesting nor bidirectional. To do so, we need to precompute
to which box we can jump from the boxed set~$\Gamma$; specifically, we need
to know the \emph{first interesting box} and the \emph{first bidirectional box}. 
We thus define $\fib(g)$ (resp., $\fbb(g)$) for a gate~$g$ as the first
interesting box (resp.,
bidirectional box) seen in the \emph{preorder traversal} of~$T$ where we first
visit the box
of~$g$, then recursively traverse its left subtree, and last traverse its right subtree.
We then extend these definitions to~$\Gamma$ by:
\begin{align}
  \fib(\Gamma) &\;\;=\;\; \min_{g \in \Gamma} \fib(g) \label{eq:fibdef}\\
  \fbb(\Gamma) &\;\;=\;\; \textsc{lca}\big(\{\fbb(g)
  \mid g \in \Gamma\}\big)\label{eq:fbbdef}
\end{align}
where the min operator is according to the preorder traversal and
$\textsc{lca}$ denotes the least common ancestor of a set of boxes in
the tree of boxes.

Intuitively, the first bidirectional box will tell us where to jump,
and the first interesting box will compensate the time spent 
visiting bidirectional boxes in the enumeration (as these boxes do not otherwise
allow us to make progress). In addition to \fib{} and
\fbb, when jumping from a box $B$ to a descendant box $B'$, we will also
need to know the $\cup$-reachability relation $R(B', B)$ from the $\cup$-gates of~$B'$
to the $\cup$-gates of~$B$, i.e., a special case of what we have defined in
Section~\ref{sec:nodupes}, where all $\cup$-gates of~$B$ appear.
Hence, our preprocessing will consider each box $B$ and compute
the $\cup$-reachability relation
$R(B', B)$ for every box~$B'$ to which we may jump from~$B$.

\paragraph*{Index Structure}
We summarize all information that needs to be computed by formally defining
the index structure:

\begin{definition}
  \label{def:index}
  The \emph{index structure} $I(C)$ of a structured complete DNNF
  $C$ consists of the following, for each box~$B$:
  \begin{itemize}
    \item For each $\cup$-gate $g \in B$, the first interesting box \fib($g$)
      of~$\Gamma$ and the reachability relation $R(\fib(g), \boxf(g))$
    \item For each boxed set $\Gamma \subseteq B$ with $1 \leq \card{\Gamma} \leq
      2$, the first bidirectional box $\fbb(\Gamma)$ of~$\Gamma$,
      and the
      reachability relation $R(\fbb(\Gamma), \boxf(\Gamma))$
    \item Letting $\mathcal{B}$ be the set of boxes of the form $\fib(g)$ or
      $\fbb(g)$ for $g$ a $\cup$-gate of~$B$, letting $\mathcal{B'} =
      \{\textsc{lca}(B_1, B_2) \mid B_1, B_2 \in \mathcal{B}\}$ 
      (hence $\mathcal{B'}
      \supseteq \mathcal{B}$), we precompute $\mathcal{B'}$ and the
      linear order implied by preorder traversal over $\mathcal{B'}$.
  \end{itemize}
\end{definition}

At first glance, the index seems weaker than what we need, because we will want
to determine $\fib(\Gamma)$ and $\fbb(\Gamma)$ for boxed sets
$\Gamma$ of arbitrary size. However,
Equation~\eqref{eq:fibdef} implies that $\fib(\Gamma)$ can be evaluated from
$\fib(g)$ for individual gates $g \in \Gamma$, using the fact that we have
precomputed $\min$.
The same is true for $\fbb$ and boxed sets~$\Gamma$ of size at most
two, thanks to Equation~\eqref{eq:fbbdef} and the following elementary fact about
least common ancestors:

\begin{observation}
  \label{obs:lca2}
  For any set $\mathcal B$ of boxes, the least common ancestor of~$\mathcal B$
  is the minimal box $B$ in the preorder traversal that is
  a least common ancestor of two (possibly equal) boxes $B_1$ and $B_2$
  of~$\mathcal B$. Formally:
$\textsc{lca}(\mathcal B)=\min \{\textsc{lca}(B_1,B_2) \mid B_i \in \mathcal B \} \label{eq:lca}
$
\end{observation}

We now show how to compute the index structure:
\begin{lemmarep}\label{lem:precompute}
  Given a structured complete DNNF 
  circuit $C$ with \mbox{$v$-tree} $\calT$, we can compute
  $I(C)$ in time $O(\card{\calT} \times w^4)$, where $w$ is the width of $C$.
\end{lemmarep}

\begin{proofsketch}
  We compute the first interesting boxes $\fib(g)$ for all $\cup$-gates $g$ by
  an easy bottom-up processing of~$C$, and we do the same for all
  $\fbb(\Gamma)$: here we rely on Observation~\ref{obs:lca2} to know that, when we recursively compute
  $\fbb(\Gamma')$ for a boxed set in a child box, we can do so from the
  $\fbb(\Gamma'')$ for $\Gamma'' \subseteq \Gamma'$ with $\card{\Gamma''} \leq
  2$.

  We compute the reachability relations by considering all boxes~$B$
  bottom-up and 
  computing $R(B', B)$ for all descendant boxes~$B'$ where this is required: we
  show that we can always do so from the relation $R(B'', B)$ for $B''$ the child
  of~$B$ in the direction of~$B'$, which is easy to compute from the wires, and from the relation $R(B', B'')$ which we
  argue must have been computed when considering~$B''$. The complexity is
  $O(w^3)$, which is bounded by the complexity of computing $R(B', B'') \circ
  R(B'', B)$ with the na\"ive join algorithm.
\end{proofsketch}

\begin{proof}
  We can compute the first interesting boxes for all $\cup$-gates in time
  $O(\card{\calT} \times w)$ by a bottom up traversal of the $\cup$-gates of the
  circuit using the following equation:
  \begin{equation}
    \fib(g)\;\;=\;\;\begin{cases}
      \boxf(g) & \text{if $g$ has a non-$\cup$ input} \\
      \min_{(g',g) \in W_\cup} \fib(g') & \text{otherwise}
    \end{cases} \label{eq:fib}
  \end{equation}
  where relation $W_\cup$ is the restriction of $W$ to $\cup$-gates. 

  Likewise, the first bidirectional boxes of at most two $\cup$-gates can be computed in time $O(\card{\calT} \times w^2)$ by
  \begin{equation}
    \fbb(\Gamma) \;\;=\;\; \begin{cases}
      \textsc{undef} & \text{if } \Gamma=\emptyset\\
      \boxf(\Gamma) & \text{if $\Gamma$ is bidirectional}\\
      \fbb(\{g' \mid (g',g) \in W_\cup, g \in \Gamma\}) \hspace{-1cm}\null& \hspace{1.2cm}\text{otherwise}
    \end{cases}\label{eq:fbb}
  \end{equation}
  where we say that $\Gamma$ is \emph{bidirectional} if some gate $g_L \in \Gamma$ has
  some input in the left child box of $\boxf(\Gamma)$ and some gate
  $g_R \in \Gamma$ has some input in the right child box of
  $\boxf(\Gamma)$ (note that we may take $g_L = g_R$).
  Observe that in the third case we call \fbb{} on a boxed set for a child box: 
  while this set may
  have cardinality $>2$, remember that we can evaluate \fbb{} from the values
  computed for the child box, simply by 
  applying Observation~\ref{obs:lca2} the definition of
  $\fbb$ in Equation~\eqref{eq:fbbdef}.

  As for reachability relations, remember that we want to compute $R(\fib(g),
  B)$ and $R(\fbb(\Gamma), B)$ for each $\cup$-gate $g \in B$ and for each boxed
  set $\Gamma \subseteq B$ such that $1 \leq \card{\Gamma} \leq 2$. Let us call
  the \emph{target boxes of~$B$} the boxes $B'$ for which we want to compute
  $R(B', B)$. First observe that:
  \begin{itemize}
    \item If $B = B'$ then $R(B', B) = \{(g, g) \mid g \in B\}$
    \item If $B'$ is a child of~$B$, then $R(B', B)$ is easily computed
    from~$W_\cup$
  \item  If $B'$ is a descendant of~$B$ but not a child, we have
  \begin{equation}
    R(B',B)\;\;=\;\;R(B',B'') \circ R(B'',B)\;, \label{eq:M}
  \end{equation}
  where $B''$ is the child of
  $B$ in the direction of $B'$.
  \end{itemize}
  The crucial observation is that, in the last case,
  we must already have precomputed $R(B', B'')$ when processing the child~$B''$
  of~$B$, i.e., $B'$ is a target box of~$B''$. But indeed:
  \begin{itemize}
    \item If $B' = \fib(g)$ for some $\cup$-gate $g$ of~$B$, then if $B' \neq
      B$, the equation for \fib{} ensures that we must have $B' = \fib(g'')$ for
      some $\cup$-gate $g'' \in B''$, so $B'$ was a target box of~$B''$.
    \item If $B' = \fbb(\Gamma)$ for some boxed set $\Gamma$, then if $B' \ne
      B$, the equation for \fbb{} ensures that we must have $B' = \fbb(\Gamma'')$
      for some boxed set $\Gamma''$ of~$B''$. Now using
      Observation~\ref{obs:lca2} we know that we have $B' = \fbb(\Gamma''_2)$ for
      some subset $\Gamma''_2 \subseteq \Gamma''$ of size at most~$2$,
      justifying that $B'$ is a target box of~$B''$.
  \end{itemize}
  
  To bound the complexity, let us first analyze how many reachability relations
  we have to compute per box: first there are  $w$ for all possible
  $\fib(g)$ and another $w$ for all possible $\fbb(g)$. For the 
  least common ancestors of all $\fbb(g)$
  at first glance it looks like we have to compute up to $w^2$ many
  reachability relations for each box. But thankfully, the set of possible least common
  ancestors is in fact of linear size. That is because the set
    $\{\;B \;\mid\; B=\textsc{lca}(\mathcal B'), \;\mathcal B' \subseteq \mathcal B \;\}$
  is of size at most $\card{\mathcal B}$. Therefore it suffices to
  overall compute at most $3w$ many relations for each box.
  Thus altogether we compute $O(\card{\calT} \times w)$ many
  relations, which takes time $O(\card{\calT} \times w^4)$ altogether, as
  each relation can be computed in time $O(w^3)$ using the na{\"i}ve
  join algorithm.
\end{proof}

\begin{algorithm}[tb]
  \caption{Box enumeration}\label{alg:boxenum}
  \begin{algorithmic}[1]
    \Procedure{box-enum}{$\Gamma$}
      \State \Call{b-enum}{$\boxf(\Gamma),\{(g,g) \mid g \in \Gamma\}$}
    \EndProcedure
    \Procedure{b-enum}{$B,R$} \label{alg:boxenum:rec}
      \State $B_1 \gets \fib(\pi_1(R))$ \Comment{first interesting box}
    \label{alg:boxenum:firstline}
      \State $R_1 \gets R(B_1,B) \circ R$ \Comment{$R(B_1,B)$ is in the
      index}
      \State \textbf{output} $R_1$\label{alg:boxenum:output} \Comment{relation
      to $B_1$}
      \State $B_L \gets \leftb(B_1)$; $R_L \gets R(B_L, B_1) \circ R_1$
      \If{$R_L \neq \emptyset$}
        \Call{b-enum}{$B_L,R_L$} \label{alg:boxenum:left}\Comment{left
        subtree of~$B_1$}
      \EndIf
      \State $B_R \gets \rightb(B_1)$; $R_R \gets R(B_R, B_1) \circ R_1$
      \If{$R_R \neq \emptyset$}
        \Call{b-enum}{$B_R,R_R$} \label{alg:boxenum:right}\Comment{right subtree of~$B_1$}
      \EndIf
      \State $B' \gets B$; $B \gets \fbb(\pi_1(R))$ \Comment{jump to the $1^\text{st}$ bidir. box}
      \While{$B$ is defined and is a strict ancestor of $B_1$}
        \State $R \gets R\big(B, B'\big) \circ R$  \Comment{$R(B,B')$ is in the index}
        \State $B_R \gets \rightb(B)$; $R_R \gets R(B_R, B) \circ R$
        \State \Call{b-enum}{$B_R,R_R$} \Comment{right subtree of $B$} \label{alg:boxenum:right2}
        \State $B' \gets \leftb(B)$; $R \gets R(B', B) \circ R$ \Comment{left child}
        \State $B \gets \fbb(\pi_1(R))$ \Comment{$1^\text{st}$ bidir. box}
      \EndWhile
    \EndProcedure
  \end{algorithmic}
\end{algorithm}

\paragraph*{Efficient Implementation of \textsc{box-enum}}
We now present the algorithm for efficient 
enumeration of \Call{box-enum}{$\Gamma$}:

\begin{lemmarep}
  \label{lem:boxenum}
  Given a 
  structured complete DNNF
  $C$ and the index structure $I(C)$,
  we can implement \textsc{box-enum} with delay $O(w^3)$, where $w$
  is the width of $C$.
\end{lemmarep}
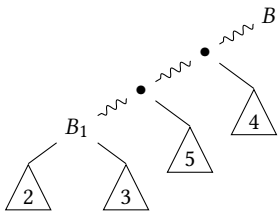
\begin{proofsketch}
  The algorithm \Call{box-enum}{$\Gamma$} to perform the enumeration for an
  arbitrary boxed set~$\Gamma$ 
  is in
  Algorithm~\ref{alg:boxenum}. Each call of the recursive algorithm in
  Line~\ref{alg:boxenum:rec} receives the relation $R(B, \Gamma)$ for
  some box~$B$ called the \emph{current box}, and it
  is expected to output the relations $R(B', \Gamma)$ for interesting boxes~$B'$ in
  the subtree of~$B$.

   \begin{figure}
    \begin{tabular}{c@{\hspace{2em}}b{3.5cm}}
    \begin{tikzpicture}[
      snake/.style={decorate,decoration={snake,amplitude=.5mm,segment
        length=1.5mm,post length=0mm}},
      triangle/.style = { draw, isosceles triangle, isosceles triangle apex angle=60, shape border rotate=90, inner sep=0pt, minimum size=.6cm, anchor=south, isosceles triangle stretches=true},
      level distance=5mm,
      sibling distance=17mm,
      ]
      \small
      \node {$B$}
      child {
        node {$\bullet$} edge from parent [snake]
        child {
          node {$\bullet$} edge from parent [snake]
          child {
            node {$B_1$} edge from parent [snake]
            child[level distance=11mm, child anchor=north,sibling distance=13mm] {
              node[triangle] {2}
            }
            child[level distance=11mm, child anchor=north,sibling distance=13mm] {
              node[triangle] {3}
            }
          }
          child[level distance=11mm, child anchor=north,sibling distance=13mm] {
            node[triangle]{5}
          }
        }
        child[level distance=11mm, child anchor=north,sibling distance=13mm] {
          node[triangle]{4}
        }
      }
      child[missing];
    \end{tikzpicture} & $B_1$ is the first interesting box and
                        $\bullet$ indicates bidirectional boxes.

                        All interesting boxes of
    the subtrees indicated by triangles are visited in the indicated
    order. 
    \end{tabular}
    \caption{Sketch of the box tree of assignment
      circuits }\label{fig:enumorder}
  \end{figure}

  In Figure~\ref{fig:enumorder}, we sketched the order in which
  boxes are enumerated, starting with the first interesting box $B_1$
  (output in Line~\ref{alg:boxenum:output}), then all descendants of
  $B_1$ (recursive calls in lines~\ref{alg:boxenum:left}
  and~\ref{alg:boxenum:right}), and then right children of bidirectional
  boxes on the path to $B_1$ (recursive call in the loop). By the
  definition of bidirectional boxes, this enumerates all interesting
  boxes.

  It is easy to show that the delay is
  $O(w^3)$, which stems from the computation of relational
  composition using na{\"i}ve joins; the main subtlety is that we need to modify
  Algorithm~\ref{alg:boxenum} slightly to ensure that each last recursive call
  is tail-recursive, to avoid delays when unwinding the recursion
  stack.
\end{proofsketch}
\begin{proof}
  The algorithm \Call{box-enum}{$\Gamma$} to perform the enumeration for an
  arbitrary boxed set~$\Gamma$ 
  is depicted as
  Algorithm~\ref{alg:boxenum}. Each call of the recursive algorithm in
  Line~\ref{alg:boxenum:rec} receives the relation $R(B, \Gamma)$ for
  some box~$B$ called the \emph{current box} and 
  is expected to output the relations $R(B', \Gamma)$ for interesting boxes~$B'$ in
  the subtree of~$B$. 

  In Line~\ref{alg:boxenum:output}, the algorithm outputs the relation
  $R(B_1,\Gamma)$ for the first interesting box $B_1$, so as to immediately make
  progress. Afterwards it does
  recursive calls in lines~\ref{alg:boxenum:left}
  and~\ref{alg:boxenum:right} that output all interesting boxes below
  the first interesting box if there are any. Finally, the algorithm walks down all
  bidirectional boxes on the path from $B$ to~$B_1$ and does recursive calls
  for the right children of these bidirectional boxes in
  Line~\ref{alg:boxenum:right2} to enumerate all interesting boxes whose
  preorder traversal number is greater than the last interesting box enumerated
  in the subtree rooted at~$B_1$. By the definition of bidirectional
  boxes, we thus enumerate all interesting boxes. We have sketched the
  order in which the boxes are enumerated in
  Figure~\ref{fig:enumorder}.

  All relations $R(B_2,B_3)$ that are used in the algorithm are either part
  of the index structure $I(C)$ or are the identity ($B_2 = B_3$) or between a
  box and a child box ($B_2$ is a child of~$B_3$). Therefore, all relational
  compositions can be computed in time $O(w^3)$ using the na{\"i}ve join algorithm.

  We show now that Algorithm~\ref{alg:boxenum} enumerates with
  constant delay, neglecting a small issue with the call stack that we
  discuss afterwards. The most important observation is that, by definition of
  bidirectional boxes, each
  recursive call will produce some output, and will do so after time
  at most $O(w^3)$, namely, the time needed to identify the first interesting
  box and produce the corresponding output.
  Then the time until we do the next recursive call
  (which itself will produce output in time $O(w^3)$) is also bounded by
  $O(w^3)$. Hence, the delay is dominated by the time spent computing the joins of the
  relations $R$.

  The only subtlety in the delay analysis concerns the call stack.
  Indeed, its depth can be as large as the depth of $C$, so cleaning up the
  stack might take too much time. To avoid this problem, we need to ensure
  that between two outputs we do not have to clean up too many stack
  frames, which we do by modifying our code to apply tail recursion elimination
  as we now explain. We modify the procedure \textsc{b-enum} so that, during
  each call to the procedure, just before we do a recursive call, we test
  whether the rest of the current execution of the procedure will be making
  another recursive call.
  For example, before doing the
  recursive call in Line~\ref{alg:boxenum:right2}, we have to check
  whether there is another bidirectional box that we need to visit in
  the next iteration of the while loop.
  Adding these tests does not impact the delay. Now, if we notice that a
  recursive call is the last one in the current execution of the procedure, we
  perform tail recursion, i.e., we do the recursive call by setting the argument
  of \textsc{b-enum} and jumping to Line~\ref{alg:boxenum:firstline} without
  adding anything to the call stack. It is clear that this change does not
  modify what the algorithm computes, as whenever we do this we have checked
  that the rest of the current execution of the procedure will not be making any
  more recursive calls; nor does the change make the delay worse.

  To understand why the algorithm is in delay $O(w^3)$ after this change,
  observe that, with the modified algorithm, we no longer need to clean up more
  than one stack frame before we can produce the next output. Indeed, if we
  clean up a stack frame, then we know that the call to which we return will be
  making another recursive call. By our analysis, it does so after a delay of at
  most~$O(w^3)$, and then we know that the call produces an output after an
  additional delay of at most~$O(w^3)$. Hence it is indeed the case that the
  overall delay in is $O(w^3)$, which concludes the proof.
\end{proof}

\paragraph*{Putting it Together}
We now state our main result about the complexity of enumerating the set of
assignments captured by a boxed set in a complete structured DNNF circuit.
Before we do so, however, we point out a small optimization trick that allows us
to bring the complexity in the width~$w$ of the circuit from~$O(w^3)$ down to $O(w^\omega)$, where
$2 \leq \omega \leq 3$ is an exponent for Boolean matrix multiplication, i.e., a constant such that the product of two $r$-by-$r$
Boolean matrices can be computed in time $O(r^\omega)$. The best possible value
for $\omega$ is an open question, 
with the best known bound being $\omega <
2.3728639$, see~\cite{Gall14a}. Observe that, in 
Theorem~\ref{thm:enumnodupes}, Lemma~\ref{lem:precompute} and
Lemma~\ref{lem:boxenum}, the complexity bottleneck is 
to compute
expressions of the form $R \circ R'$ for relations $R$ and $R'$ over sets of
size $\leq w$, with all other operations having complexity $O(w^2)$ at most.
We have used the na\"ive join algorithm to bound this by $O(w^3)$,
but we can instead represent these relations as Boolean matrices and
use any matrix multiplication algorithm to compute the product in 
$O(w^\omega)$.

This leads to our final enumeration result 
on set circuits:

\begin{theorem}
  \label{thm:masterenum}
  Let $\omega$ be an exponent for the Boolean matrix multiplication problem.
  Given any complete structured DNNF~$C$ of width~$w$ with its v-tree $\calT$ and
  structuring function, we can preprocess~$C$ 
  in $O(\card{\calT} \times w^{\omega+1})$  to be able, given any boxed set $\Gamma$, to
  enumerate the assignments of $\S(\Gamma)$ with delay $O(\card{S} \times
  w^\omega)$ for each produced assignment~$S$.
\end{theorem}

\begin{proof}
  We preprocess $C$ using Lemma~\ref{lem:precompute}, and we then perform the
  enumeration using Theorem~\ref{thm:enumnodupes} with the efficient
  implementation of \textsc{box-enum} given in Lemma~\ref{lem:boxenum};
  modifying the algorithms to perform matrix multiplication in time $O(w^\omega)$ instead
  of using the na\"ive join algorithm.
\end{proof}

\section{Updates and Balancing}
\label{sec:updates}
We have shown our circuit construction result (Lemma~\ref{lem:buildcircuit})
and enumeration result (Theorem~\ref{thm:masterenum}). We will put them together
to show our enumeration results for automata and queries over trees. However, before this we need to explain how
we can handle updates efficiently, i.e., how we can recompute the circuit
(Lemma~\ref{lem:buildcircuit}) and the index of the enumeration structure
(Lemma~\ref{lem:precompute}) whenever the underlying tree is modified.

The crucial insight is that the circuit in
Lemma~\ref{lem:buildcircuit} is computed \emph{bottom-up} on the input tree~$T$, 
and the precomputation in
Lemma~\ref{lem:precompute} is also performed bottom-up on the tree of boxes
whose structure is isomorphic to~$T$. Hence, whenever $T$ is updated at some
node~$n$, we can modify the circuit and the index accordingly by recomputing
everything bottom-up starting at node~$n$. The complexity of this
process will be linear in the \emph{height} of~$T$. This is why, in this
section, we will want to work on trees that are \emph{balanced}, i.e., whose
height is logarithmic in their size: this is what will guarantee that updates
can be handled in logarithmic time.

As trees are in general not balanced, our technique will be to code the input
tree as binary
balanced trees. We will also use this as a way to allow arbitrary unranked trees
as input (not just binary trees), which is more convenient because binary trees do not behave well under edit
operations (e.g., adding or deleting a single leaf).
Specifically, given the input unranked $\Lambda$-tree $T$ and $\Lambda,\X$-TVA~$A$ running on
unranked trees,
we will encode~$T$ to a balanced binary tree~$T'$ over a different alphabet
$\Lambda'$, and we will translate
the automaton~$A$ in polynomial time to a $\Lambda',\X$-TVA~$A'$, while ensuring
that
$A$ and $A'$ have the same satisfying assignments. Our balanced binary tree
formalism will further ensure that, whenever an update is performed on~$T$, we
can update $T'$ (and keep it balanced) to reflect the change, and we can
efficiently update the circuit~$C$ and the index for this update.

In this section, we first present our model for the input unranked tree~$T$ and the
edit operations that we allow on it. We then present our formalism for the
automaton~$A$, which runs on unranked trees, and the notion of
\emph{tree hollowings} to describe which kinds of updates can happen on the
balanced binary tree~$T'$: intuitively, whenever we apply an edit operation on
$T$, then we will be able to update~$T'$ in logarithmic time with a tree
hollowing; and we will show that we can update the circuit and index in the same
time bound. Then, we formalize the notion of encoding unranked trees to binary
trees and of faithfully translating automata, and we state the result
from~\cite{Niewerth18} which explains how this can be efficiently performed. 
This allows us state our main enumeration results for automata and queries in the next section.

\paragraph*{Edit Operations on Unranked Trees.}
We first present the language of edit operations that we allow on the
\emph{unranked} tree which is the input to our enumeration scheme.
Fixing a set $\X$ of variables, we will define
an $\X$-valuation of an unranked $\Lambda$ tree~$T$ as a function $\nu$ mapping
each node $n$ of~$T$ to a set $\nu(n) \subseteq \X$. Note that a valuation of an unranked
tree annotates all its nodes, not just the leaf nodes.

The update operations that we allow on unranked trees are leaf insertions,
leaf deletions, and relabelings. More precisely:
\begin{definition}
  \label{def:edits}
  Given an unranked $\Lambda$-tree $T$, a node $n$ of~$T$ and a label $l \in \Lambda$, we allow the following edit operations:
  \begin{itemize}
  \item $\textsc{delete}(n)$: remove $n$ from $T$  (only if $n$ is a leaf)
  \item $\textsc{insert}(n,l)$: insert an $l$-node as first child of~$n$
  \item $\textsc{insert$_R$}(n,l)$: insert an $l$-node as right sibling of $n$
  \item $\textsc{relabel}(n,l)$: change the label of $n$ to $l$
  \end{itemize}
  For any edit operation $\tau$, we call $\tau(T)$ the
  resulting tree.
\end{definition}

\paragraph*{Automata on Unranked Trees.}
Following our use of unranked binary trees, we must also extend the definition
of TVAs to work on unranked trees. Our automaton model for unranked trees are stepwise tree automata extended with variables. 
Stepwise tree automata where introduced in~\cite{CarmeNT-rta04}. We use the formalism from~\cite{MartensN-jcss07a}, as it is closer to our tree model.
A \emph{$\Lambda,\X$-TVA on unranked $\Lambda$-trees} for the variable set~$\X$ 
is a tuple $A = (Q, \iota, \delta, F)$,
where $Q$ are the states,
$\iota \subseteq \Lambda \times 2^{\X} \times Q$ is the \emph{initial
  relation}, $\delta \subseteq Q \times Q \times Q$ is the
\emph{transition relation}, and $F \subseteq Q$ are the 
\emph{final states}.
Given a $\Lambda$-tree $T$ and an $\X$-valuation $\nu$ of~$T$, we again
define a \emph{run} of~$A$ on~$T$ as a function $\rho \colon T\to Q$ satisfying some
conditions. However, compared to binary trees, the use of $\iota$ and $\delta$ is
different: $\iota$ assigns a set of possible initial states to every state (not only to
leaves). For inner nodes, $\delta$ then consumes the states of children state by state ,
just as a word automaton reads its input symbol by symbol. 
The assigned state of a node is the state after having read the states of all children.

To
formalize this, let us see~$\delta$ as a function $\delta\colon 2^Q \times Q \to 2^Q$
by setting $\delta(Q',q')=\cup_{q \in Q'} \{q'' \mid (q,q',q'')
\in \delta\}$, and let us inductively define the function 
$\delta^* \colon 2^Q \times Q^* \to 2^Q$ 
by $\delta^*(Q',\epsilon)\colonequals Q'$ and
$\delta^*(Q',q_1\dots q_n)\colonequals \delta^*(\delta(Q',q_1),q_2\dots q_n)$.
A \emph{run} $\rho \colon T \to Q$ of~$A$ on~$T$ must then satisfy
the following: For every node $n$ with children $n_1,\dots,n_m$,
    we have $\rho(n) \in
    \delta^*(\iota(\lambda(n),\nu(n)),\rho(n_1)\rho(n_2)\dots\rho(n))$.
    Note that TVAs on unranked trees now read annotations at all nodes (and not only at leaves),
    as per the definition of valuations.

The definitions of accepting runs, satisfying valuations and assignments on a
$\Lambda,\X$-TVA on unranked trees is the same as for binary trees: note that
assignments now consist of singletons of the form $\langle Z: n\rangle$ where
$n$ can be any node of~$T$, not just a leaf.
Incidentally, observe that any $\Lambda,\X$-TVA on binary trees can clearly be converted to a
$\Lambda,\X$-TVA on unranked trees which accepts exactly the same trees.

\paragraph*{Tree Hollowings.}
We now formalize our language of updates on \emph{binary trees}, which we call tree
hollowings. Intuitively, when performing an update on the input unranked
tree~$T$, we want to reflect it on its balanced binary representation~$T'$.
Our encoding scheme will ensure that the edit operations of
Definition~\ref{def:edits} can be performed as \emph{tree hollowings}
whose \emph{trunks} have logarithmic size. Here are the relevant definitions:

\begin{definition}
  Let $T'$ be a binary $\Lambda'$-tree. Two nodes of~$T'$ are \emph{incomparable} if
  neither is a descendant of the other (in particular, they must be different).
  An \emph{antichain} in~$T'$ is a set of nodes in~$T'$ that are pairwise
  incomparable.

  A \emph{tree hollowing} $H=(T'',\eta)$ of~$T'$ consists of a
  $(\Lambda' \cup \{\square\})$-tree~$T''$ called the \emph{trunk}, in which all
  internal nodes must have a label in~$\Lambda'$, and an injective function
  $\eta$ from the leaves of~$T''$ with label~$\square$ to~$T'$, such
  that the image of~$\eta$ is a antichain of~$T'$. The \emph{result}
  $\result{H}$ of the tree hollowing $H$ is the $\Lambda'$-tree obtained
  by taking $T''$ and replacing each $\square$-labeled leaf $n$ of~$T''$
  by the subtree of~$T'$ rooted at~$\eta(n)$.
\end{definition}

Intuitively, a tree hollowing describes how to build a new tree while reusing
disjoint subtrees of the original tree. We will require that each update to the
original tree~$T$ should be translatable in logarithmic time to a tree
hollowing of the binary balanced representation~$T'$---so in particular the trunk will have logarithmic size, even
though the result of the hollowing will not.

The reason why hollowings are a good update language is because our
constructions are strictly bottom up. Thus,
given a binary $\Lambda'$-tree~$T'$ for which we have computed an assignment
circuit~$C$ for some $\Lambda',\X$-TVA $A$, and given the index $I(C)$, we can 
follow a tree hollowing $H=(T'',\eta)$ of~$T'$ to update the circuit to an
assignment circuit $\result{C}$ of~$A$ on~$\result{H}$ and to update the index to
$I(\result{C})$. Formally:

\begin{lemma}
  \label{lem:hollowing}
  Given any $\Lambda'$-tree $T'$ and $\Lambda',\X$-TVA $A'$ with state space~$Q'$, given an assignment
  circuit $C$ of~$A'$ on~$T'$ which is a structured complete DNNF of width
  $\card{Q'}$ and given the index structure
  $I(C)$, given any tree hollowing $H' = (T'', \eta)$ of $T'$, we can compute in time
  $O(\card{T''} \times \card{Q'}^{\omega+1})$
  a circuit $\result{C}$ and the index structure $I(\result{C})$
  such that $\result{C}$ is a structured complete DNNF of width $\card{Q'}$ which
  is an assignment circuit of~$A'$ on~$\result{T'}$, and $I(\result{C})$ is the
  index structure for~$\result{C}$.
\end{lemma}
\begin{proof}
  We use the circuit construction from Lemma~\ref{lem:buildcircuit} to
  compute a box $B_n$ for each node $n$ in $T''$ that is not labeled
  $\square$ (total time $O(\card{T''}\times\card{A'})$) and we thus compute $\result{C}$,
  which has size $O(\card{T''} \times\card{A'})$. Afterwards we use the
  computation from the proof of
  Lemma~\ref{lem:precompute} (using efficient matrix multiplication as explained
  in the proof of Theorem~\ref{thm:masterenum})
  to compute the index structure for the new boxes in total time $O(\card{T'}
  \times w^{\omega+1})$.
\end{proof}

\paragraph*{Encoding Unranked Trees in Balanced Binary Trees.}
We will now explain how we can encode an unranked tree~$T$ into a balanced binary
tree~$T'$
on which update operations can be reflected, i.e., every update
on~$T$ in the language of Definition~\ref{def:edits} translates to an update
of~$T'$ that can be represented by a tree hollowing whose trunk has logarithmic
size. It will be more convenient to formalize the encoding by describing the
\emph{decoding function} which decodes a binary tree to an
unranked tree. Specifically,
a \emph{tree encoding scheme} for the tree alphabet~$\Lambda$
consists of a tree alphabet $\Lambda'$ and of a function $\omega$ defined on some
subset of the binary $\Lambda'$-trees (called the \emph{well-formed trees})
and which decodes
any
such binary $\Lambda'$-tree $T'$ to an unranked $\Lambda$-tree $T$ and to a bijection
$\phi_{T'}$ from the
leaves of~$T'$ to the nodes of~$T$. We require $\omega$ to be surjective, i.e.,
every unranked $\Lambda$-tree has some preimage in~$\omega$. Now, given a set
$\X$
of variables, a $\Lambda,\X$-TVA $A$ and a binary $\Lambda,\X'$-TVA $A'$, we say that $\omega$
is \emph{$A,A'$-faithful} if for any $\Lambda$-tree $T$, for any preimage $T'$
of~$T$ in~$\omega$, letting $\phi_{T'}$ be the bijection from the leaves of~$T'$
to the nodes of~$T$, for any $\X$-valuation $\nu$ of~$T$, we have that $A$
accepts $T$ under $\nu$ iff $A'$ accepts $T'$ under the valuation $\nu \circ
\phi_{T'}$.

\begin{toappendix}
  In this appendix, we give some more explanations about the proof of
  Lemma~\ref{lem:treebal} using~\cite{Niewerth18}. We start with some
  prerequisites about forest algebra terms.

\paragraph*{Forest Algebra Terms}
A \emph{forest algebra pre-term} is a term in the free forest algebra as
defined in~\cite{Niewerth18} that does not use the empty forest and empty context.
We repeat the definition for self-containedness.

A \emph{$\Lambda$-forest} is an ordered list of~$\Lambda$-trees.
A \emph{$\Lambda$-context} is a $(\Lambda \cup \{\hole\})$-forest, where the
special ``hole'' label $\hole$ is applied to no internal node and to exactly
one leaf.

  A \emph{forest algebra pre-term} on an alphabet $\Lambda$ is a binary tree whose internal nodes are labeled
  $\oplus_{HH}$ (for forest concatenation), $\oplus_{HV}$, $\oplus_{VH}$ (for concatenation of forest and context or vice versa), $\odot_{VV}$ (for context application between contexts) and $\odot_{VH}$ (for context application on a forest), whose
  leaves are labeled with $a_t$ for $a \in \Lambda$ (for a node labeled
  $a$), or $a_\hole$ for $a \in \Lambda$ (for a context node labeled $a$), and
  where we require that some typing constraints are respected. Specifically,
  each node of the tree is typed as a forest or as a context, with $a_t$ being a
  forest and $a_\hole$ being a context, and the type of inner nodes is defined by
  induction:
  \begin{itemize}
  \item $\oplus_{HH}$ both inputs must be forests, and the result is a forest;
  \item $\oplus_{HV}$ the left input must be a forest, the right input
    must be a context, and the result is a context;
  \item $\oplus_{VH}$ the right input must be a forest, the left input must be a
    context, and the result
    is a context;
  \item $\odot_{VV}$ both inputs must be contexts, and the result is a context;
  \item $\odot_{VH}$ the left input must be a context, the right input must be a
    forest, and the result is a forest.
  \end{itemize}
  We say that a forest algebra term \emph{represents} an unranked forest or context
  on the alphabet~$\Lambda$, which is defined by induction,
  preserving the invariant that a node typed as a forest represents a
  $\Lambda$-forest and a node typed as a context represents a $\Lambda$-context:
  \begin{itemize}
    \item $a_t$ represents the forest with a singleton root labeled~$a$;
    \item $a_\hole$ represents the context with a singleton root labeled~$a$ and
      having a single child labeled~$\hole$;
    \item for a $\oplus$-node, given the contexts or forests $F_1$ and $F_2$ represented by
      the first and second input, the result is~$F_1 \cup F_2$;
    \item for a $\odot$-node, given the context $F_1$ and the forest or
      context $F_2$ represented respectively by
      the first and second input, the result is obtained by
      replacing the one $\hole$-labeled node of~$F_1$ by $F_2$,
      i.e., the $\hole$-labeled node $n$ of~$F_1$ is removed and the roots of
      the trees in~$F_2$ are inserted in the list of children of the parent
      of~$n$ at the point where $n$ was, in their
      order according to~$F_2$;
    \end{itemize}
    
  A \emph{forest algebra term} is a forest algebra pre-term that represents a $\Lambda$-forest (i.e., the root has type forest), and where this forest contains
  exactly one tree.
\end{toappendix}

Our tree encoding method can then be formalized as the following result, which easily follows from~\cite{Niewerth18}:
\begin{lemmarep}
  \label{lem:treebal}
  For any tree alphabet $\Lambda$ and set~$\X$ of variables, there is an
  encoding scheme $\omega$ for~$\Lambda$ such that:
  \begin{itemize}
    \item The encoding is linear-time computable, i.e., given any
      unranked $\Lambda$-tree $T$, we can compute in linear time some
      $\Lambda'$-tree~$T'$ with $\omega(T')=T$, as well as the bijection $\phi_{T'}$.
    \item The encoded trees have logarithmic height, i.e., each
      well-formed $\Lambda'$-tree $T'$ has height in
      $O(\log \card{T'})$.
    \item We can efficiently perform updates, i.e., given any binary
      $\Lambda$-tree~$T$, preimage $T'$ of~$T$, and update $\tau$ on~$T$ in the
      language of Definition~\ref{def:edits}, we can compute in
      time $O(\log \card{T})$ a tree hollowing $H=(T'', \eta)$ of\/~$T'$
      such that $\omega(\result{H})=\tau(\omega(T'))$.
    \item We can efficiently translate automata, i.e., given any
      unranked $\Lambda,\X$-TVA $A$ with state space~$Q$, we can build in time
      $O(\card{Q}^6)$ a binary \mbox{$\Lambda',\X$-TVA} $A'$ with
      $O(\card{Q}^4)$ states and $O(\card{Q}^6)$ transitions such that
      $\omega$ is $A,A'$-faithful. Furthermore, $A'$ has a single
      accepting state.
  \end{itemize}
\end{lemmarep}

\begin{proofsketch}
  The idea is to convert unranked trees to binary trees
  that represent terms in the free
  forest algebra. Intuitively, a forest algebra term $T'$ is a binary
  $\Lambda'$-tree that describes an unranked $\Lambda$-tree:
  the leaves of~$T'$ correspond to the nodes of~$T$, and each internal node
  corresponds to a forest computed from its children 
  by one of two operations: forest concatenation or context application.
  In~\cite[Section~3]{Niewerth18}, it is described
  how we can convert an unranked tree~$T$ in linear time to a balanced term~$T'$ that
  represents~$T$, and
  how updates on~$T$ can be reflected on~$T'$.
  Although this is not stated explicitly
  in~\cite{Niewerth18}, the resulting updates on~$T'$ can be
  described as tree hollowings of logarithmic size. We can also easily show that
  an automaton $A$ on unranked $\Lambda$-trees can be converted in PTIME to an
  automaton $A'$ on binary $\Lambda'$-trees.
\end{proofsketch}
\begin{proof}
  We let $\Lambda'$ be the alphabet of forest algebra terms over
  alphabet~$\Lambda$, and
  we define $\omega$ to be the function that maps forest algebra terms
  over alphabet $\Lambda$ to the unranked $\Lambda$-trees that they represent.
  We artificially
  restrict $\omega$ to be only defined on those forest algebra terms
  that have logarithmic height in order to satisfy the second
  condition.

  It is easy to see that every tree can be encoded as a forest algebra
  term. In~\cite{Niewerth18}, it is shown that for each tree there
  exists a forest algebra term of logarithmic height and that we can
  efficiently perform updates. Although this is not stated explicitly
  in~\cite{Niewerth18}, all updates of forest algebra terms can be
  described as tree hollowings of logarithmic size. The initial update
  on the forest algebra term is of constant size and all successive
  rotations done for rebalancing are either performed on the path from
  the updated node to the root or performed on a direct child of a
  node from this path, so the result of the update on the forest
  algebra is indeed representable as a trunk of logarithmic size to
  which we connect subtrees of the original forest algebra term.

  There is a natural bijection between leaves of forest algebra terms
  and nodes in the represented tree.
  
  The last point is to explain how we convert an unranked
  $\Lambda,\X$-TVA $A=(Q,\delta,F)$ to a binary $\Lambda',\X$-TVA
  $A'=(Q',\delta',F')$. We do so as follows, where we assume
  \mbox{w.l.o.g.} that we have added to~$Q$ some special states
  $q_0,q_f$ such that
  $\delta \cap (\{q_0\} \times Q \times \{q_f\}) = \{q_0\} \times F
  \times \{q_f\}$. This will help us later to identify the accepting runs of
  $A$.
  \begin{align*}
    \Lambda' \quad=&\quad \{a_t \mid a \in \Lambda\} \;\; \cup \;\; \{a_\square \mid a \in \Lambda\} \;\;\cup\;\; \{\oplus_{HH},\oplus_{HV},\oplus_{VH},\odot_{VV}, \odot_{VH}\} \\
    Q' \quad=&\quad Q^2 \;\; \cup \;\; (Q^2)^2 \\
    F' \quad=&\quad \{(q_0,q_f)\} \\
    \iota' \quad=&\quad \big\{\; \big(a_t,\Y,(q_1,q_2)\big) \;\;\big|\;\; (q_1,p,q_2) \in
    \delta \text{ for some } p \in \iota(a,\Y) \;\big\} \;\;\cup\\
                  &\quad \big\{\;\big(a_\hole,\Y,((q_1,q_2),(q_3,q_4))\big) \;\;\big|\;\;
                  (q_1,q_4,q_2) \in \delta, q_3 \in \iota(a,\Y) \;\big\} \\
    \delta'_{\oplus_{HH}} \quad=&\quad \Big\{\;\Big((q_1,q_2),\; (q_2,q_3), \;(q_1,q_3)\Big) \;\;\Big|\;\; q_1.q_2,q_3 \in Q \;\Big\} \\
    \delta'_{\oplus_{HV}} \quad=&\quad \Big\{\;\Big((q_1,q_2),\;\big((q_2,q_3),(q_4,q_5)\big),\big((q_1,q_3),(q_4,q_5)\big)\Big) \;\;\Big|\;\; q_1,\dots,q_5 \in Q\;\Big\}\\
    \delta'_{\oplus_{VH}} \quad=&\quad \Big\{\;\Big(\big((q_1,q_2),(q_3,q_4)\big),\;(q_2,q_5),\;\big((q_1,q_5),(q_3,q_4)\big)\Big)\;\;\Big|\;\; q_1,\dots,q_5 \in Q \;\Big\} \\
    \delta'_{\odot_{VV}} \quad=&\quad \Big\{\;\Big(\big((q_1,q_2),(q_3,q_4)\big),\;\big((q_1,q_2),(q_3,q_4)\big),\;\big((q_1,q_2),(q_3,q_4)\big)\Big) \;\;\Big|\;\; q_1,\dots,q_6 \in Q\;\Big\}\\
    \delta'_{\odot_{VH}}\quad=&\quad
    \Big\{\;\Big(\big((q_1,q_2),(q_3,q_4),\;(q_3,q_4),\;(q_1,q_2)\big)\Big) \;\;\Big|\;\; q_1,\dots,q_4 \in Q \;\Big\}
  \end{align*}

  The construction of $A'$ is directly derived from the definition of
  the transition algebra as given in~\cite[Section~4]{Niewerth18}. The
  resulting automaton has $O(\card{Q}^4)$ states and $O(\card{Q^6})$
  transitions and can be computed in time $O(\card{Q}^6)$.

  We have to prove faithfulness. Rather then repeating the definition
  of transition algebra form~\cite{Niewerth18}, we give a direct
  description on how $A'$ captures runs of $A$. This is similar to the
  proof of Lemma~9 in~\cite{Niewerth18}.

  In the following, we let $T'$ be an input tree for $A'$ and
  $T=\omega(T')$ and let $\psi$ be a function mapping nodes of $T'$ to
  the forest or context represented by the forest algebra pre-term
  rooted at $n$.

  It can be verified by induction that the state $q$ assigned to a
  node $n$ in a run of $A'$ satisfies the following conditions that
  describe the partition of $Q'$ into states $Q^2$ for nodes of type
  forest and $(Q^2)^2$ for nodes of type context.
  \begin{itemize}
  \item If $\psi(n)$ is a forest then $q=(q_1,q_2) \in Q^2$ such that
    there exists a run $\rho$ on $\psi(n)$ such that we have
    $\delta(q_1,\rho(n_1)\cdots\rho(n_m))=q_2$, where $n_1,\dots,n_m$
    are the roots of $\psi(n)$.
  \item If $\psi(n)$ is a context then $q=((q_1,q_2),(q_3,q_4))$ such
    that---after replacing the hole in $\psi(n)$ with a forest that
    allows the transition from $q_1$ to $q_2$---there exists a run
    $\rho$ such that $\delta(q_3,\rho(n_1)\cdots\rho(n_m))=q_4$, where
    $n_1,\dots,n_m$ are the roots of $\psi(n)$.
  \end{itemize}

  The induction base case is at the leaves of $T'$. If $n$ is labeled
  $a_t$ with $a \in \Lambda$, $A'$ nondeterministically guesses a pair
  of states such that $\delta(q_1,p)=q_2$, where $p$ is a state that
  $A$ could assign to a leaf with label $a$, i.e., it has to be an
  initial state that is possible for label $a$ and the given variable
  assignment at the node.

  If $n$ is labeled $a_\square$, then $A'$ guesses a pair of states
  $((q_1,q_2),(q_3.q_4))$ such that the forest that will be inserted
  into the hole allows a transition from $q_3$ to $q_4$. Furthermore
  $A$ has to allow a transition from $q_1$ to $q_2$ when reading the
  state assigned to $\varphi_{T'}(n)$.

    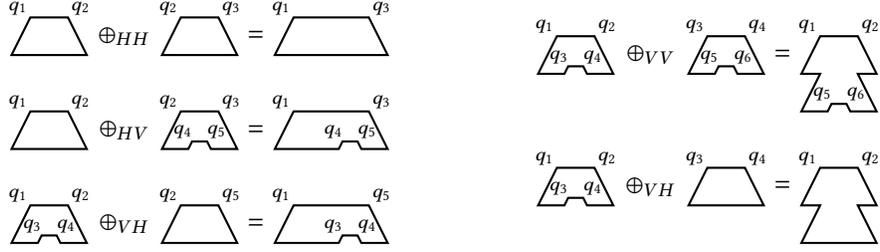
\begin{figure}
    \begin{tikzpicture}
      \tikzstyle{sl}=[inner sep=1pt]

      \draw[thick]  (0,0) coordinate (x1l) -- ++(5mm,0mm) coordinate (x1r) -- ++(2.5mm,-5mm) -- ++(-10mm,0mm) -- cycle;
      \draw[thick]  (2,0) coordinate (x2l) -- ++(5mm,0mm) coordinate (x2r) -- ++(2.5mm,-5mm) -- ++(-10mm,0mm) -- cycle;
      \draw[thick]  (3.5,0) coordinate (x3l) -- ++(10mm,0mm) coordinate (x3r) -- ++(2.5mm,-5mm) -- ++(-15mm,0mm) -- cycle;
      
      \node at (1.25,-0.25) {$\oplus_{HH}$};
      \node at (3,-0.25) {$=$};
      
      { \scriptsize
        \node[sl,anchor=south east] at (x1l) {$q_1$};
        \node[sl,anchor=south west] at (x1r) {$q_2$};

        \node[sl,anchor=south east] at (x2l) {$q_2$};
        \node[sl,anchor=south west] at (x2r) {$q_3$};

        \node[sl,anchor=south east] at (x3l) {$q_1$};
        \node[sl,anchor=south west] at (x3r) {$q_3$};
      }

      \draw[thick]  (0,-1.25) coordinate (x1l) -- ++(5mm,0mm) coordinate (x1r) -- ++(2.5mm,-5mm) -- ++(-10mm,0mm) -- cycle;
      \draw[thick]  (2,-1.25) coordinate (x2l) -- ++(5mm,0mm) coordinate (x2r) -- ++(2.5mm,-5mm) -- ++(-3.5mm,0mm) -- ++(-.5mm,1mm) coordinate (x2br) --  ++(-2mm,0mm) coordinate (x2bl) -- ++(-.5mm,-1mm) -- ++(-3.5mm,0mm) -- cycle;
      \draw[thick]  (3.5,-1.25) coordinate (x3l) -- ++(10mm,0mm) coordinate (x3r) -- ++(2.5mm,-5mm) -- ++(-3.5mm,0mm) -- ++(-.5mm,1mm) coordinate (x3br) -- ++(-2mm,0mm) coordinate (x3bl) -- ++(-.5mm,-1mm) -- ++(-8.5mm,0mm) -- cycle;
      
      \node at (1.25,-1.5) {$\oplus_{HV}$};
      \node at (3,-1.5) {$=$};
      
      { \scriptsize
        \node[sl,anchor=south east] at (x1l) {$q_1$};
        \node[sl,anchor=south west] at (x1r) {$q_2$};

        \node[sl,anchor=south east] at (x2l) {$q_2$};
        \node[sl,anchor=south west] at (x2r) {$q_3$};
        \node[sl,anchor=south east] at (x2bl) {$q_4\!$};
        \node[sl,anchor=south west] at (x2br) {$\!q_5$};

        \node[sl,anchor=south east] at (x3l) {$q_1$};
        \node[sl,anchor=south west] at (x3r) {$q_3$};
        \node[sl,anchor=south east] at (x3bl) {$q_4\!$};
        \node[sl,anchor=south west] at (x3br) {$\!q_5$};
      }

      \draw[thick]  (0,-2.5) coordinate (x1l) -- ++(5mm,0mm) coordinate (x1r) -- ++(2.5mm,-5mm) -- ++(-3.5mm,0mm) -- ++(-.5mm,1mm) coordinate (x1br) --  ++(-2mm,0mm) coordinate (x1bl) -- ++(-.5mm,-1mm) -- ++(-3.5mm,0mm) -- cycle;
      \draw[thick]  (2,-2.5) coordinate (x2l) -- ++(5mm,0mm) coordinate (x2r) -- ++(2.5mm,-5mm) -- ++(-10mm,0mm) -- cycle;
      \draw[thick]  (3.5,-2.5) coordinate (x3l) -- ++(10mm,0mm) coordinate (x3r) -- ++(2.5mm,-5mm) -- ++(-3.5mm,0mm) -- ++(-.5mm,1mm) coordinate (x3br) -- ++(-2mm,0mm) coordinate (x3bl) -- ++(-.5mm,-1mm) -- ++(-8.5mm,0mm) -- cycle;
      
      \node at (1.25,-2.75) {$\oplus_{VH}$};
      \node at (3,-2.75) {$=$};
      
      { \scriptsize
        \node[sl,anchor=south east] at (x1l) {$q_1$};
        \node[sl,anchor=south west] at (x1r) {$q_2$};
        \node[sl,anchor=south east] at (x1bl) {$q_3\!$};
        \node[sl,anchor=south west] at (x1br) {$\!q_4$};

        \node[sl,anchor=south east] at (x2l) {$q_2$};
        \node[sl,anchor=south west] at (x2r) {$q_5$};

        \node[sl,anchor=south east] at (x3l) {$q_1$};
        \node[sl,anchor=south west] at (x3r) {$q_5$};
        \node[sl,anchor=south east] at (x3bl) {$q_3\!$};
        \node[sl,anchor=south west] at (x3br) {$\!q_4$};
      }

      \draw[thick]  (7,-.25) coordinate (x1l) -- ++(5mm,0mm) coordinate (x1r) -- ++(2.5mm,-5mm) -- ++(-3.5mm,0mm) -- ++(-.5mm,1mm) coordinate (x1br) --  ++(-2mm,0mm) coordinate (x1bl) -- ++(-.5mm,-1mm) -- ++(-3.5mm,0mm) -- cycle;
      \draw[thick]  (9,-.25) coordinate (x2l) -- ++(5mm,0mm) coordinate (x2r) -- ++(2.5mm,-5mm) -- ++(-3.5mm,0mm) -- ++(-.5mm,1mm) coordinate (x2br) --  ++(-2mm,0mm) coordinate (x2bl) -- ++(-.5mm,-1mm) -- ++(-3.5mm,0mm) -- cycle;
      \draw[thick]  (10.5,-.25) coordinate (x3l) -- ++(5mm,0mm) coordinate (x3r) -- ++(2.5mm,-5mm) -- ++(-2.5mm,0mm) -- ++(2.5mm,-5mm) -- ++(-3.5mm,0mm) -- ++(-.5mm,1mm) coordinate (x3br) -- ++(-2mm,0mm) coordinate (x3bl) -- ++(-.5mm,-1mm) -- ++(-3.5mm,0mm) -- ++(2.5mm,5mm) -- ++(-2.5mm,0mm) -- cycle;
      
      \node at (8.25,-.5) {$\oplus_{VV}$};
      \node at (10,-.5) {$=$};
      
      { \scriptsize
        \node[sl,anchor=south east] at (x1l) {$q_1$};
        \node[sl,anchor=south west] at (x1r) {$q_2$};
        \node[sl,anchor=south east] at (x1bl) {$q_3\!$};
        \node[sl,anchor=south west] at (x1br) {$\!q_4$};

        \node[sl,anchor=south east] at (x2l) {$q_3$};
        \node[sl,anchor=south west] at (x2r) {$q_4$};
        \node[sl,anchor=south east] at (x2bl) {$q_5\!$};
        \node[sl,anchor=south west] at (x2br) {$\!q_6$};

        \node[sl,anchor=south east] at (x3l) {$q_1$};
        \node[sl,anchor=south west] at (x3r) {$q_2$};
        \node[sl,anchor=south east] at (x3bl) {$q_5\!$};
        \node[sl,anchor=south west] at (x3br) {$\!q_6$};
      }

      \draw[thick]  (7,-2) coordinate (x1l) -- ++(5mm,0mm) coordinate (x1r) -- ++(2.5mm,-5mm) -- ++(-3.5mm,0mm) -- ++(-.5mm,1mm) coordinate (x1br) --  ++(-2mm,0mm) coordinate (x1bl) -- ++(-.5mm,-1mm) -- ++(-3.5mm,0mm) -- cycle;
      \draw[thick]  (9,-2) coordinate (x2l) -- ++(5mm,0mm) coordinate (x2r) -- ++(2.5mm,-5mm) -- ++(-10mm,0mm) -- cycle;
      \draw[thick]  (10.5,-2) coordinate (x3l) -- ++(5mm,0mm) coordinate (x3r) -- ++(2.5mm,-5mm) -- ++(-2.5mm,0mm) -- ++(2.5mm,-5mm) -- ++(-10mm,0mm) -- ++(2.5mm,5mm) -- ++(-2.5mm,0mm) -- cycle;
      
      \node at (8.25,-2.25) {$\oplus_{VH}$};
      \node at (10,-2.25) {$=$};
      
      { \scriptsize
        \node[sl,anchor=south east] at (x1l) {$q_1$};
        \node[sl,anchor=south west] at (x1r) {$q_2$};
        \node[sl,anchor=south east] at (x1bl) {$q_3\!$};
        \node[sl,anchor=south west] at (x1br) {$\!q_4$};

        \node[sl,anchor=south east] at (x2l) {$q_3$};
        \node[sl,anchor=south west] at (x2r) {$q_4$};

        \node[sl,anchor=south east] at (x3l) {$q_1$};
        \node[sl,anchor=south west] at (x3r) {$q_2$};
      }
  \end{tikzpicture}
  \caption{Visualisation of the monoid operations. Forests are depicted
    as trapezoids and contexts as trapezoids with a cutout. States
    from $A$ that are memorized in states of $A'$ are indicated at the
    corresponding corners.}\label{fig:monoid}
\end{figure}

  For the induction step, one has to look at the transitions allowed
  by $A'$ on inner nodes, i.e., on nodes with a label from
  $\{\oplus_{HH},\oplus_{HV},\oplus_{VH},\odot_{VV},\odot_{VH}\}$.
  There $A'$ verifies deterministically whether the nondeterministic
  choices done at the leaves are consistent. In the following $n$ will
  always be an inner node with left child $n_L$ and right child $n_R$.
  We depicted a sketch of all five monoid operations in
  Figure~\ref{fig:monoid}. There we also indicate the state names used
  in the definition of $\delta'$.  

  If $n$ has label $\oplus_{HH}$, $A'$ has to check that the state
  $q_2$ after reading the (roots of the) forest $\psi(n_L)$ has to be
  the same as the state before reading the (roots of the) forest
  $\psi(n_R)$.  Furthermore, $A'$ propagates the states $(q_1,q_2)$
  upwards, where $q_1$ is the state before reading $\psi(n_L)$ and
  $q_3$ is the state after reading $\psi(n_R)$.

  If $n$ has label $\odot_{VH}$, $A'$ checks that the forest
  $\psi(n_R)$ actually allows the transition that was guessed for the
  hole at some leaf with some label $a_\square$.

  In the two described cases, the resulting type is forest. In the
  remaining three cases, where the resulting type is context, $A'$
  additionally has to propagate the guess for the hole upwards, such
  that it can be verified later, when the automaton reaches a node
  labeled $\odot_{VH}$.

  The correctness of this verification at inner nodes can be verified
  using the definition of $\delta'$. At last, the automaton has to
  check that all the guesses are not only consistent, but belong to an
  accepting run of $A$. This is reflected by our choice of $F'$ that
  verifies that $A$ can do a transition from $q_0$ to $q_f$ when
  reading the state assigned to the root of $\omega(T)$. According to our
  assumption, this implies that the root is assigned a state from $F$
  and thus $A$ accepts $\omega(T)$.

  We note that the translation preserves runs in the sense that for
  every unique run of $A$ on $\omega(T)$, there is a unique run of
  $A'$ on $T'$. Especially, if $A$ is deterministic or unambiguous,
  then $A'$ is unambiguous. Furthermore, for nondeterministic automata
  $A$, the number of runs is preserved.
\end{proof}

\section{Main Results}
\label{sec:main}
We can now present our main results by combining our results about tree
balancing and hollowing updates (Lemma~\ref{lem:treebal} and
Lemma~\ref{lem:hollowing}) with our circuit construction and  enumeration results
(Lemma~\ref{lem:buildcircuit} and Theorem~\ref{thm:masterenum}).
The first phrasing of our main result deals with 
the enumeration of the satisfying assignments to TVAs under
updates, while ensuring the right complexity bounds, and in particular
guaranteeing tractable combined complexity:

\begin{theoremrep}
  \label{thm:tva}
  Let $\omega$ be an exponent for the Boolean matrix multiplication problem.
  Given an unranked $\Lambda,\X$-TVA $A$ with state space~$Q$ and an unranked $\Lambda$-tree $T$, we can
  enumerate the satisfying assignments of $A$ on $T$ with preprocessing time
  $O(\card{T} \times \card{Q}^{4(\omega+1)})$, update time
  $O(\log(\card{T}) \times \card{Q}^{4(\omega+1)})$, and delay
  $O(\card{Q}^{4\omega} \times \card{S})$, where $S$ is the produced assignment.
\end{theoremrep}

\begin{proof}
  Let $\omega$ be the encoding scheme of Lemma~\ref{lem:treebal}.
  In the preprocessing phase, we apply Lemma~\ref{lem:treebal} to compute a 
  $\Lambda'$-tree $T'$ such that $\omega(T') = T$, and the bijection
  $\phi_{T'}$, in time $O(\card{T})$. We also know that the height of~$T'$ is
  logarithmic. We also translate the automaton~$A$ in time $O(\card{Q}^6)$ to a
  binary $\Lambda',\X$-TVA $A'$ with state space $Q'$ such that $\card{Q'} =
  O(\card{Q}^4)$ and such that $A'$ has $O(\card{Q}^6)$-transitions and has only
  one final state.
  We then use Lemma~\ref{lem:homogenize} to process~$A'$ and ensure that it is
  homogenized; the construction clearly ensures that $A'$ then has exactly two
  final states: one final $0$-state $q_{f,0}$ and one final $1$-state $q_{f,1}$.
  Further, we know that $\omega$ is $A,A'$-faithful, meaning in
  particular that for any $\X$-valuation $\nu$ of the unranked tree~$T$,
  letting $\nu' \colonequals \nu \circ \phi_{T'})$, we have that $A$ accepts $T$
  under $\nu$ iff $A'$ accepts $T'$ under~$\nu'$.

  We now apply Lemma~\ref{lem:buildcircuit}
  to construct in time $O(\card{T'} \times \card{A'})$, i.e., $O(\card{T} \times
  \card{Q}^6)$, a structured complete DNNF $C$ which is an assignment circuit
  of~$A$ and~$T$, such that the width~$w$ of~$C$ is~$\card{Q'}$, i.e.,
  $O(\card{Q}^4)$, and its depth is $O(\height(T'))$,
  i.e., $O(\log(\height(T)))$. We will then perform enumeration using
  Theorem~\ref{thm:masterenum}, which includes a preprocessing phase in time
  $O(\card{C} \times w^\omega)$, i.e., $O(\card{T} \times w^{4\omega+6})$, to
  compute the index $I(C)$. Hence, the preprocessing time is as we claimed.
  
  After
  the preprocessing, letting $g_{f,1} \colonequals \phi(n, q_{f,1})$ where $n$ is
  the root of~$T'$, we can
  enumerate the assignments of~$\S(\{q_{f,1}\})$ with delay $O(\card{S} \times
  w^{omega})$, i.e., $O(\card{S} \times \card{Q}^{4\omega})$, where $S$ is the
  produced assignment, which is the time bound that we claimed. To see why this
  is correct, observe that by definition of an assignment circuit, this enumerates
  precisely the assignments corresponding to valuations~$\nu'$ such that $A'$
  has a run mapping the root of~$T'$ to~$q_{f,1}$, and as~$q_{f,1}$ is the only
  final $1$-state of~$A'$ it is exactly the set of assignments corresponding to
  non-empty valuations~$\nu'$ such that $A'$ accepts~$T'$ under~$\nu'$. Last, we
  handle the case of the empty valuation by considering $g_{f,0} \colonequals
  \phi(n,q_{f,0})$ for~$n$ the root of~$T'$, which must clearly be either a
  $\top$-gate or~$\bot$-gate, and we produce the empty assignment iff $g_{f,1}$
  is a $\top$-gate: by the same reasoning this produces the empty assignment iff
  $A'$ accepts $T'$ under the empty valuation. Thus, we correctly enumerate the
  set of assignments corresponding to valuations $\nu'$ such that $A'$ accepts
  $T'$ under~$\nu$, i.e., the satisfying assignments of~$A'$ on~$T'$. As $\omega$ is
  $A,A'$-faithful, this enumerates exactly the satisfying assignments of~$A$ on~$T$. Hence,
  the enumeration is correct.

  Now, whenever an update is performed on~$T$, we know by
  Lemma~\ref{lem:treebal} that we can compute in time $O(\log{T})$ a tree
  hollowing $H = (T'', \eta)$ of~$T'$ such that $\omega(\result{H}) =
  \tau(\omega(T))$. Now, we know that we can reflect this change on~$C$ and on
  the index $I(C)$ to obtain a new circuit $\result{C}$ and index
  $I(\result{C})$ such that $\result{C}$ is a structured complete DNNF of
  width $\card{Q'}$, i.e., the same as~$C$, which is an assignment circuit
  of~$A'$ on the new $\Lambda'$-tree~$\result{H}$, and $I(\result{C})$ is the
  index structure for~$\result{C}$. This update takes time
  $O(\card{T'} \times \card{Q'}^{\omega+1})$, i.e., $O(\card{T}
  \times \card{Q}^{4\omega+1})$. Thus the update time is as we claimed, which
  concludes the proof.
\end{proof}

Thanks to Theorem~\ref{thm:tva}, we can now obtain as a corollary our main
result on enumeration for monadic second-order logic (MSO). Recall that MSO
is a logic 
that extends first-order logic
on $\Lambda$-trees defined on a signature
featuring the edge relation of the tree, the order relation among siblings, and unary predicates for the node
labels. MSO extends first-order logic by adding the ability to quantify over sets. 
Given a $\Lambda$-MSO formula $\query(X_1, \ldots, X_n)$ 
and a $\Lambda$-tree $T$
where all free variables are
second-order, letting $\X = \{X_1, \ldots, X_n\}$,
the \emph{satisfying valuations} of~$\query$ on~$T$ are the $\X$-valuations $\nu$ of~$T$
such that $T$
satisfies the Boolean $\Lambda$-MSO formula $\query$ where the
predicates $X_1, \ldots, X_n$ are interpreted according to~$\nu$. Similarly to
what we did for TVAs, we represent any $\X$-valuation $\nu$ as an
\emph{assignment} $\alpha(\nu) \colonequals \{\langle Z: n\rangle \mid Z \in
\nu(n)\}$, and the 
\emph{satisfying assignments} of~$\query$ on~$T$ are the images by~$\alpha$ of the
satisfying valuations.

Our main result for MSO enumeration is then the following; it
improves on the
earlier results for MSO enumeration on trees under
updates~\cite{losemann2014mso,amarilli2018enumeration,Niewerth18}. Note that
unlike Theorem~\ref{thm:tva}, this result only considers \emph{data complexity},
i.e., complexity in the input tree (and in the produced assignments), assuming
that the MSO query is fixed.
\begin{corollary}
  \label{cor:mso1}
  For any fixed tree alphabet $\Lambda$, for any fixed \mbox{$\Lambda$-MSO} query $\query$ with
  free second-order variables,
  given an unranked \mbox{$\Lambda$-tree}~$T$, 
  after preprocessing~$T$ in linear time,
  we can enumerate the satisfying assignments of $\query$ with delay linear in each produced
  assignment, and we can handle updates to~$T$ in logarithmic time in~$T$.
\end{corollary}

\begin{proof}
  Let $\X$ be the variables of~$\query$.
  We see $\query$ as a Boolean query on the tree signature $\Lambda \cup 2^{\X}$,
  and rewrite $\query$ to a tree automaton on $\Lambda \cup 2^{\X}$ using the well-known
  result by Thatcher and Wright~\cite{thatcher1968generalized}: see also \cite[Appendix
E.1]{amarilli2017circuit_arxiv}. The
  automaton can equivalently be seen as an unranked $\Lambda,\X$-TVA, so we can
  enumerate its assignments on~$T$ using Theorem~\ref{thm:tva}.
\end{proof}

In the common case of an MSO formula $\query(x_1, \ldots, x_n)$
with free \emph{first-order} variables,
Corollary~\ref{cor:mso1} implies that we can enumerate the satisfying assignments on a
tree~$T$,
with the delay being constant,
i.e., independent from~$T$.

\begin{corollaryrep}
  \label{cor:mso2}
  For any fixed tree alphabet $\Lambda$, for any fixed $\Lambda$-MSO query $\query$ with
  free first-order variables,
  given an unranked \mbox{$\Lambda$-tree}~$T$, 
  after preprocessing~$T$ in linear time,
  we can enumerate the satisfying assignments of $\query$ with constant delay, 
  and we can handle updates to~$T$ in logarithmic time in~$T$.
\end{corollaryrep}

\begin{proof}
  We do the standard rewriting of $\query(x_1, \ldots, x_n)$ to an MSO query
  $\query'(X_1,
  \ldots, X_n)$ with free second-order variables where we add for each~$i$ a
  conjunct asserting that $X_i$ is a singleton (e.g., $\exists x ~ X_i(x) \land
  (\forall x y ~ X_i(x)
  \land X_i(y) \rightarrow x = y)$) and we add an existential quantification
  $\exists x_i ~ X_i(x_i)$, then we reuse the body of~$\query$. This rewriting is
  independent of~$T$, so runs in constant time. Now, we use
  Corollary~\ref{cor:mso1} to enumerate the satisfying assignments of~$\query'$
  on~$T$. The definition of~$\query'$ ensures that each satisfying assignment has
  cardinality exactly~$n$, so the enumeration proceeds in constant time. This
  produces the desired result because there is a clear bijection from the
  satisfying assignments of~$\query'$ on~$T$ to the
  answer tuples of~$\query$ on~$T$, which we can apply to each satisfying assignment
  in constant time.
\end{proof}

Note that enumerating satisfying assignments also allows us
to enumerate the \emph{answer tuples} $(a_1, \ldots, a_n) \in T^n$
such that $T$ satisfies $\query(a_1, \ldots, a_n)$, as we can simply translate each
satisfying assignment in linear time to an answer tuple.

\paragraph*{Results on Words}
We conclude the section by presenting consequences of our results for
\emph{words}.
As words are a special case of trees, our results on trees imply results for the
enumeration on words of the satisfying assignments of word automata, with the
ability to efficiently handle updates to the underlying word. This can be used
in the context of \emph{document spanners}~\cite{FaginKRV15} for information
extraction, allowing us to efficiently enumerate the matches of a document
spanner represented as a sequential extended VA~\cite{florenzano2018constant},
and to update the enumeration structure when the word changes.

Given a set $\Lambda$ of labels, we call a \emph{$\Lambda$-word} a finite
sequence~$w$ of letters from~$\Lambda$.
Given a variable set~$\X$, an \emph{$\X$-valuation}
of~$w$ is a function $\nu: \{1, \ldots, \card{w}\} \to 2^{\X}$, and the
corresponding \emph{$\X$-assignment} $\alpha(\nu)$ is the set $\{\langle Z:
n\rangle \mid Z \in \nu(n)\}$.
On words we allow the usual local edits: (1) inserting a character, (2)
deleting a character, and (3) replacing a character.

We next present the formalism that we use to write queries on words,
which is analogous to \emph{extended sequential variable automata}~\cite{florenzano2018constant}.
A \emph{word variable automaton} WVA on $\Lambda$-words for variable
set~$\X$ (or $\Lambda,\X$-WVA) is a tuple $A = (Q, \delta, I, F)$,
where $Q$ is a set of \emph{states}, $I$ is the set of initial states,
$\delta \subseteq Q \times \Lambda \times 2^{\X} \times Q$ is
the \emph{transition relation}, and $F \subseteq Q$ is the set of
\emph{final states}.
Note that, like an unranked TVA, a WVA can read variables at any
position.

When working with WVAs rather than TVAs, the
translation to balanced binary trees can be done with a better 
complexity.
Specifically, the following is obvious from the details of the construction of $A'$
in Lemma~\ref{lem:treebal}.
\begin{corollaryrep}
  Lemma~\ref{lem:treebal} holds for words instead of trees and a WVA
  $A$ instead of an unranked TVA as input, with the difference that
  the binary tree automaton $A'$ can be constructed in time
  $O(\card{Q}^3)$, has $O(\card{Q}^2)$ states and $O(\card{Q}^3)$
  transitions, where $Q$ is the state space of~$A$.
\end{corollaryrep}
\begin{proof}
  We interpret strings as forests, where each tree has exactly one
  node.  We can use the same balancing schema that happens to work
  exactly like AVL trees, if the only operation is concatenation
  ($\oplus_{HH}$).

  We also reuse the automaton construction with the difference that we
  can drop everything that is related to contexts.  Especially, we can
  use $Q'=Q^2$, as there are are no nodes of type
  context. Furthermore, we only need $\delta'_{\oplus_{HH}}$ from the
  definition of $\delta'$, as there are no other operators.
  \mbox{W.l.o.g.} we assume that $A$ has only one initial state $q_0$ and one final state $q_F$.

  This can be easily achieved by adding a new initial state $q_0$ ans
  final state $q_F$, such that $q_0$ gets all outgoing transitions
  from all initial states and $q_F$ gets all incoming transitions from
  all existing final states. Afterwards we use $q_0$ and $q_F$ as
  single initial and final state, respectively.

  We use $(q_0,q_F)$  as the single final state of $A'$.
\end{proof}

We can conclude the following theorem.
\begin{theorem}
  \label{thm:wva}
  Given a $\Lambda,\X$-WVA $A$ with state space~$Q$
  and $\Lambda$-word $w$, we can
  enumerate the satisfying assignments 
  with preprocessing time
  $O(\card{w} \times \card{Q}^{2(\omega+1)})$, update time
  $O(\log(\card{w}) \times \card{Q}^{2(\omega+1)})$, and delay
  $O(\card{Q}^{2\omega} \times \card{v})$, for $v$ the current valuation.
\end{theorem}

In terms of document spanners,
this theorem is the analogue of the result of~\cite{amarilli2019constant} which
showed that we can efficiently enumerate the results of an extended sequential
nondeterministic VA on an input word, and it extends this result to handle
updates to the word in logarithmic time. However, in exchange for the support
for updates, the complexity in the automaton is less favorable, with a higher
polynomial degree (due to the need to balance the word); and the memory
usage of Theorem~\ref{thm:wva} is not constant like it was
in~\cite{amarilli2019constant}. We also note that our results in this paper do
\emph{not} recapture the results for \emph{non-extended VAs}
from~\cite{amarilli2019constant}: we believe that our techniques here would
extend to such automata, but that it would require some changes to the tree
automaton model.

\section{Lower Bound}
\label{sec:lower}
In this section, we will show that the logarithmic update time of
Section~\ref{sec:main} is optimal up to a doubly logarithmic factor. In fact, we will show that either the update time or the enumeration delay of any MSO enumeration algorithm on trees must be $\Omega\left(\frac{\log(n)}{\log \log(n)}\right)$. In particular, in contrast to the settings in~\cite{berkholz2017answering2,berkholz2017answering}, there is no algorithm with constant update time for MSO queries on trees even if we allow the enumeration delay to be slightly superconstant.

To show our lower bound, we rely heavily on a result from~\cite{AlstrupHR98}, so let us introduce some notation from there: consider a tree $T$ on~$n$ nodes in which some of the nodes are \emph{marked} while the others are \emph{unmarked}. For every node $v$ of $T$, let $\pi(v)$ denote the path from~$v$ to the root of $T$. The \emph{existential marked ancestor query} is, given a node~$v$, to decide if $v$ has a marked ancestor, i.e., there is a marked node on $\pi(v)$. 
An \emph{update} in the marked ancestor problem is an operation that marks or unmarks a node of $T$. An algorithm for the marked ancestor problem is an algorithm that maintains a data structure that allows updates and marked ancestor queries. A main result of~\cite{AlstrupHR98} is then the following:

\begin{theorem}\label{thm:alstruphr}
 Consider an algorithm to solve the marked ancestor problem and let $t_u$ be a bound on the time to handle an update and
  $t_q$ be a bound on the time to answer a query.
  Then we have $t_q = \Omega\left(\frac{\log(n)}{\log(t_u\log(n))}\right)$.
\end{theorem}
We remark that Theorem~\ref{thm:alstruphr} is unconditional
and holds for the standard model of unit cost RAMs with logarithmic word size; for different word sizes the result is slightly different. We also remark that Theorem~\ref{thm:alstruphr} even stays true if all runtimes are assumed to be amortized, so in a non-worst case setting. Finally, Theorem~\ref{thm:alstruphr} makes no assumption whatsoever about the runtime of a potential preprocessing phase. So even with generous runtime in the preprocessing, the result holds.

We now show that we get a lower bound for MSO query enumeration easily from Theorem~\ref{thm:alstruphr}.

\begin{theoremrep}\label{thm:lower}
 There is an MSO query $\query$ on trees such that 
 any  enumeration algorithm for $\query$ under relabelings with update time $\hat{t}_u$ and enumeration delay $\hat{t}_e$ has
 \[\max(\hat{t}_u, \hat{t}_e) \;\;\ge\;\; \Omega\left(\frac{\log(n)}{\log\log(n)}\right).\]
\end{theoremrep}

\begin{proofsketch}
 We reduce from the existential marked ancestor problem. To this end, we consider trees in which nodes can have three labels: marked, unmarked, or special. Fix the query $\query(x)$ that selects all special nodes that have a marked ancestor: we can easily write it in MSO.
 
 We give an algorithm for the marked ancestor problem: We start with the input tree $T$, i.e., a tree without any node marked special. To answer the marked ancestor query for a node $v$, we label $v$ as a special node, enumerate the answer to $\query$, and make $v$ non-special again. Finally we return 'yes' if and only if we enumerated any answer to~$\query$.
 
 To see that this algorithm is correct, observe that $v$ is the only special
  node in $T$ when evaluating $\query$. So either we enumerate $v$ or nothing
  depending on if $v$ has a marked ancestor, and the answer we give is correct.
  Now,
  Theorem~\ref{thm:alstruphr} applies to the marked
  ancestor queries, i.e., to $2 \hat{t}_u + \hat{t}_e$, from which we can
  mathematically derive our claimed bound.
\end{proofsketch}
\begin{proof}
  The beginning of the proof is as explained in the sketch: we now give
 the computation for the lower bound. The runtime of the marked ancestor queries
  as we implemented them is $t_{\query} = 2 \hat{t}_u + \hat{t}_e$. From Theorem~\ref{thm:alstruphr} we get
  \begin{align} \label{eq:lower} t_{\query} \;\;=\;\; 2 \hat{t}_u + \hat{t}_e \;\;\ge\;\; c\frac{\log(n)}{\log(t_u\log(n))} \;\;=\;\; c\frac{\log(n)}{\log(\hat{t}_u\log(n))}\end{align}
  for some constant $c$. 
  
  Now, assume first that $\hat{t}_u > \hat{t}_e$. Then we get from~\eqref{eq:lower} that
  \[3 \hat{t}_u \;\;\ge\;\; c \frac{\log(n)}{\log(\hat{t}_u\log(n))}.\]
  Now assume w.l.o.g.~that $\hat{t}_u \le \log(n)$ (otherwise we are done), then by substituting on the right-hand side we get
  \[3 \hat{t}_u \;\;\ge\;\; c \frac{\log(n)}{\log(\log(n)\log(n))} \;\;\ge\;\; \frac{\log(n)}{2\log\log(n)}\]
  which completes the proof.
  
  If $\hat{t}_e \ge \hat{t}_u$, then we get from \eqref{eq:lower} that 
  \[3 \hat{t}_e \;\;\ge\;\; c \frac{\log(n)}{\log(\hat{t}_e\log(n))}\]
  and reasoning as before completes the proof.
  \end{proof}
From Theorem~\ref{thm:lower} we get easily that the update time in Theorem~\ref{thm:tva} is optimal up to a doubly logarithmic factor even if we allow the enumeration delay to be close to logarithmic.
\begin{corollary}\label{cor:lower}
 There is an MSO query $\query$ on trees such that 
 any  enumeration algorithm for $\query$ under relabelings with enumeration delay
  $o\left(\frac{\log(n)}{\log \log(n)}\right)$ has update time
  $\Omega\left(\frac{\log(n)}{\log\log(n)}\right)$.
\end{corollary}

As Theorem~\ref{thm:alstruphr} from~\cite{AlstrupHR98}, our Theorem~\ref{thm:lower} also works for amortized time and with arbitrary preprocessing. Note also that the query $\query$ in the proof does not use the full power of MSO but can in fact be expressed in 
first-order logic with transitive closure, so our lower bound already holds for that fragment.
Note that it is not clear if we can show an analogous result for 
FO
without transitive closure: in particular, we cannot apply the results of~\cite{berkholz2017answering2} about enumeration for 
FO on bounded-degree structures, because
the lower bound in~\cite{AlstrupHR98} uses trees of unbounded degree.

\section{Conclusion}
\label{sec:conc}
We have shown an efficient algorithm to enumerate the assignments of MSO
formulae on trees under updates (relabeling, leaf insertions, leaf deletions),
with linear preprocessing in the input tree, linear delay in
each produced assignment (so constant if all free variables are first-order),
and with logarithmic update time (in the tree). Our work is the first to
match the bounds of the original results on MSO enumeration on trees without
updates, while allowing general updates on trees (refer back to
Table~\ref{tab:overview} for a comparison).
Our algorithm
is also the first to be tractable in combined complexity when the query is
given as a (generally nondeterministic) tree variable automaton, with the
preprocessing phase, enumeration delay, and update time being polynomial in the
automaton; this extends our previous results on
words~\cite{amarilli2019constant} to trees and to efficient updates.

Our results leave several directions open for future work. One question is to
improve the complexity in terms of the automaton, i.e., lowering the polynomial
degree, as we have shown to be possible in the case of queries on words. A
related question would be to perform enumeration for a more concise
automaton model, in particular allowing the automaton to represent
possible sets of captured variables more concisely, like
non-extended VAs \cite{florenzano2018constant,amarilli2019constant}: we
believe that this should be possible. We could also aim for more expressive automata models,
e.g., alternating automata or two-way automata; or other
query languages on trees, e.g., tree pattern queries.
Another open question is the support for more expressive update operations. In
the case of words, it would be natural to support bulk updates, i.e.,
moving a part of the text to a different place (see the conclusion
of~\cite{amarilli2019constant}). We believe that our techniques 
could adapt for such updates on words. As for trees, we currently do not know
how to handle updates that
split a subtree or attach a subtree.

One issue that we have not explored in the current paper is memory usage.
Constant memory usage was achieved in~\cite{amarilli2019constant} (for
nondeterministic sequential VAs) and
in~\cite{bagan2006mso,kazana2013enumeration,amarilli2019constant} (for MSO
queries with free first-order variables), but it does not hold for our
results (the memory usage may be linear in the circuit). We do not
know if we can achieve constant memory.

Finally, there is a gap of $\log\log(n)$ between our upper bound and
the lower bound, which it would be interesting to close.
Note that the marked ancestor problem in fact has an algorithm
with update complexity $O(\log(n)/\log\log(n))$, so we cannot hope to
close the gap by improving the lower bound in
Theorem~\ref{thm:alstruphr}.
It may be the case that better lower bounds can be shown for our
enumeration problem, but, going beyond $\Omega(\log(n)/\log\log(n))$
is generally considered very challenging. Indeed, the only paper that
we are aware of that achieves a lower bound of $\Omega(\log(n))$ for
any dynamic problem is~\cite{PatrascuD06}, but it does not imply a
lower bound in our setting because they allow more powerful updates
than we do.
Alternatively, it may be possible to improve our update complexity to
$O(\log(n)/\log\log(n))$, e.g., by adapting
the techniques of~\cite{AlstrupHR98}, but we have not been able to do so.

It would also be interesting to see what happens if we allow more
generous enumeration delay, say $O(\log(n))$. Can we get an algorithm
with updates that is faster than the lower bound of Corollary~\ref{cor:lower}
in that case, maybe even constant? This might be an interesting
trade-off in applications where updates are much more frequent
than enumeration queries.

\begin{acks}
  Bourhis was partially founded by the DeLTA French
  \grantsponsor{ANR}{ANR}{http://anr.fr} project (\grantnum{ANR}{ANR-16-CE40-0007}).
  Niewerth was supported by the \grantsponsor{DFG}{Deutsche
    Forschungsgemeinschaft}{http://dfg.de} (\grantnum{DFG}{MA 4938/4-1}).
\end{acks}
\clearpage

\bibliographystyle{ACM-Reference-Format}
\bibliography{main}

\end{document}